\documentclass[english]{article}
\usepackage[utf8]{inputenc}
\usepackage[T1]{fontenc}
\usepackage{lmodern}
\usepackage[a4paper]{geometry}
\usepackage[english]{babel}
\usepackage{csquotes}
\usepackage{tabularx,booktabs}

\usepackage{amsmath}
\usepackage{amssymb, amsthm}
\usepackage{stmaryrd}
\usepackage{empheq}
\usepackage{bm}
\usepackage{slashed}
\usepackage{microtype}

\usepackage{mathrsfs}

\usepackage[dvipsnames]{xcolor}
\usepackage{soulutf8}
\newcommand{\com}[1]{\textcolor{red}{[\hl{#1}]}}
\renewcommand{\com}[1]{}
\usepackage{hyperref}

\usepackage{tikz-cd}

\newtheorem{theorem}{Theorem}[section]
\newtheorem{lemma}[theorem]{Lemma}

\newtheorem{proposition}[theorem]{Proposition}
\newtheorem{definition}[theorem]{Definition}
\theoremstyle{remark}
\newtheorem*{remark}{Remark}

\newcommand{\setR}{\mathbb{R}}

\newcommand{\setN}{\mathbb{N}}
\newcommand{\setC}{\mathbb{C}}
\newcommand{\setK}{\mathbb{K}}

\newcommand{\h}{{\mathfrak \xi}}
\newcommand{\bh}{{\bar{\h}}}

\newcommand{\Sp}{\Sigma}
\newcommand{\bSp}{\Sp^*}

\newcommand{\bP}{{\bar\Psi}}
\newcommand{\bp}{{\bar\psi}}

\newcommand{\met}{{\eta}}

\newcommand{\Lie}{\mathcal L}

\newcommand{\Lag}{\mathscr{L}}

\newcommand{\ExT}{\mathsf{\Lambda}}

\renewcommand{\d}{{\operatorname d}} 

\newcommand{\dom}{\d^\omega}

\newcommand{\lmtrois}[1]{\lambda^{\m(3)}}

\newcommand{\omsix}{\omega^{(6)}}
\newcommand{\omcinq}[1]{\omega^{(5)}_{#1}}
\newcommand{\aquatre}{\alpha^{(4)}}
\newcommand{\atrois}[1]{\alpha^{(3)}_{#1}}
\newcommand{\adeux}[1]{\alpha^{(2)}_{#1}}
\newcommand{\aun}[1]{\alpha^{(1)}_{#1}}
\newcommand{\vsept}[1]{{\varpi^{(7)}_{#1}}}
\newcommand{\vhuit}[1]{{\varpi^{(8)}_{#1}}}
\newcommand{\vneuf}[1]{{\varpi^{(9)}_{#1}}}
\newcommand{\vdix}{{\varpi^{(10)}}}

\newcommand{\gneuf}{{\gamma^{(9)}}}
\newcommand{\gtrois}{{\gamma^{(3)}}}

\newcommand{\nonbc}[1]{{\widehat{#1}}}

\newcommand{\Omhor}{\Omega_{\textrm{hor}}}

\newcommand{\lr}{\lrcorner\,}

\newcommand{\vol}{\mathrm{vol}}
\newcommand{\ve}{\varepsilon}

\newcommand{\lp}{\left(}
\newcommand{\rp}{\right)}

\newcommand{\wb}[2]{\left[#1\wedge#2 \right]}
\newcommand{\fwb}[2]{\frac{1}{2}\wb{#1}{#2}}

\DeclareMathOperator{\Ric}{Ric}
\DeclareMathOperator{\Scal}{Scal}
\newcommand{\Curv}{R}

\DeclareMathOperator{\tr}{tr}

\DeclareMathOperator{\End}{End}
\DeclareMathOperator{\Iso}{Iso}

\DeclareMathOperator{\Spin}{Spin}

\DeclareMathOperator{\GL}{GL}

\DeclareMathOperator{\id}{id}

\newcommand{\Cinf}{\mathcal{C}^\infty}

\newcommand{\isom}{{\xrightarrow{\sim}}}

\renewcommand{\varphi}{\phi}
\newcommand{\tot}{\mathcal P}

\DeclareMathOperator{\SO}{SO}

\newcommand{\SOp}{\SO^{+}}
\DeclareMathOperator{\SOpL}{{\SOp_{1,3}}}
\newcommand{\Spinp}{\Spin^+}
\renewcommand{\so}{\operatorname{\mathfrak{so}}}
\newcommand{\spin}{\operatorname{\mathfrak{spin}}}
\DeclareMathOperator{\soL}{{\so_{1,3}}}
\DeclareMathOperator{\Mink}{\m}
\DeclareMathOperator{\Cl}{Cl}

\newcommand{\aux}{V}

\newcommand{\Lor}{\SOp_{1,3}}
\renewcommand{\lor}{\so_{1,3}}
\let\lbar\l
\renewcommand{\l}{\lor} 

\newcommand{\m}{{\setR^{1,3}}}
\newcommand{\ET}{{\mathcal E}}

\newcommand{\sopq}{\so_{p,q}}
\newcommand{\Rpq}{{\setR^{p,q}}}

\newcommand{\sprod}[2]{\langle #1 | #2 \rangle}

\newcommand{\lel}[3]{#1\leqslant #2 \leqslant #3}

\usepackage{breakurl}
\usepackage[isbn=false]{biblatex}
\title{Physics and Geometry from a Lagrangian: Dirac Spinors on a Generalised Frame Bundle}
\author{Jérémie Pierard de Maujouy
	\thanks{ \textit{Institut Mathématiques de Jussieu-Paris Rive Gauche (IMJ-PRG)}
	 	\textit{Université Paris Cité}, France, \newline
	 	{jeremie.pierard-de-maujouy@imj-prg.fr}
	 }
}

\date{\today}

\begin{document}

\maketitle

\begin{abstract}
We clarify the structure obtained in Hélein and Vey's proposition for a variational principle for the Einstein-Cartan gravitation formulated on a frame bundle starting from a structure-less differentiable $10$-manifold \cite{LFB}. The obtained structure is locally equivalent to a frame bundle which we term \enquote{generalised frame bundle}. In the same time, we enrich the model with a Dirac spinor coupled to the Einstein-Cartan spacetime. The obtained variational equations generalise the usual Einstein-Cartan-Dirac field equations in that they are shown to imply the usualy field equations when the generalised frame bundle is a standard frame bundle.

\end{abstract}

\tableofcontents
\section{Introduction}
\subsection{Introduction}

The theory of General Relativity models spacetime as a $4$-dimensional differentiable manifold. The gravitational field is encoded in a linear connection and provides \emph{geometrical structure} to the manifold. In the original approach, a Lorentzian metric models the gravitational potential and the gravitational field corresponds to the associated Levi-Civita connection. The variational formulation of the theory is due to Einstein and Hilbert~\cite{MTW}; the dynamical field is the metric and the Lagrangian density is simply the scalar curvature (multiplied by the pseudo-Riemannian density).

Following the development of the framework of differential geometry as well as research on General Relativity, alternative formulations were discovered.
The \emph{Palatini formulation} of General Relativity relaxes the relation between the linear connection and the metric. This allows for a first-order formulation of gravity. A variant of this formulation uses the \emph{tetrad formalism}: the metric field is replaced by a frame field which defines the metric for which it is orthonormal and an orientation. Since the metric field is a quadratic function of the frame field, the tetrad can be thought of as a \enquote{square-root of the metric}. It is non-unique and is in fact subject to a $\SOp_{1,3}$ gauge freedom.

The \emph{Einstein-Cartan theory} is the case where there is a Lorentzian metric and the connection is required to be metric but torsion is allowed. The corresponding geometry on spacetime is suitably understood in the framework of \emph{Cartan geometry}~\cite{CartanWise}.
Differences with Einstein's original theory manifest themselves when the gravitational field is coupled to spinor fields, since they act as a source for the torsion field. One simple example is the \emph{Einstein-Cartan-Dirac} theory which couples the Einstein-Cartan gravitation with a Dirac spinor. The usual variational treatment of spinor field theories in a dynamical spacetime uses a tetrad in order to allow variations of the metric structure while the coefficients of the spinor fields remain \enquote{constant}. When a tetrad is used instead of a metric, the theory is called \emph{Sciama-Kibble} theory \cite{GRSpinTors, STandFields}.

The theory we discuss in this paper stems from the idea of formulating gravitation not on spacetime itself but on the frame bundle of spacetime. This frame bundle forms a principal bundle above spacetime and General Relativity and its generalisation can be understood as gauge theories for the Poincaré group~\cite{GaugeGrav,GaugeGravBook}. As such, our theory can be likened to the Kaluza-Klein theories in which a theory on a principal bundle induces a field theory on the base spacetime manifold. More motivations come from the perspective of a quantum version of the theory. On one hand, Lurçat suggested~\cite{QFTPoincare, SpinGroupKinematics} to consider interacting theories not on Minkowski space but on the Poincaré group: this could allow for quantum particles to go out of their usually fixed \enquote{spin shell}.
On the other hand, since our theory only fixes the manifold structure of the frame bundle, it is conceivable to obtain solutions with non diffeomorphic spacetimes. This possibility, as well as its quantum implications, are left for further investigation.

In~\cite{LFB} Hélein and Vey, following~\cite{CFTRefFrames}, proposed an action defined on a (structure-less) $10$-dimensional manifold $\tot$ such that under some hypotheses a solution of the associated Euler-Lagrange equations defines:
\begin{itemize}
\item A $4$-dimensional manifold $\ET$ over which $\tot$ is fibred,
\item A Riemannian metric and an orientation on $\ET$,
\item An identification from the fibration bundle $\tot\to \ET$ to the orthonormal frame bundle of $\ET$,
\item A metric connection on $\ET$ satisfying the (Riemannian) Einstein-Cartan field equations.
\end{itemize}
As is already mentioned in~\cite{LFB}, their model can be used with a Lorentzian signature (and space-and-time orientations) but in this case both the construction of the fibration and the mechanism giving the Einstein-Cartan field equations in vacuum require additional hypotheses to work. The reason lies in the non-compactness of the (connected) Lorentz group. We argue that there is another obstruction to the construction of the fibration $\tot\to\ET$, even in Riemannian signature. Indeed in this case the $4$-dimensional base space $\ET$, obtained as an orbit space, can have singularities and is not \emph{a priori} a smooth manifold. Our second point is that their derivation of the field equations involves heavy computations on which we pretend a more geometrical perspective can shed some light.

The aim of the present paper is twofold: we provide an extension of Hélein \& Vey's model to Einstein-Cartan's gravity coupled to Dirac spinors, and we take a very different approach to the cancellation mechanism which allows the factor the obtained field equations to the underlying spacetime.


The (tentative) principal bundle structure is constructed from a nondegenerate $1$-form $\omega\oplus\alpha$ with value in the Poincaré Lie algebra decomposed as $\lor\ltimes \Mink$. 
The $1$-form obeys the following equations: 
\begin{subequations}\label{eqno:introCartForm}
\begin{align}
	\d \omega^i + \fwb{\omega}{\omega}^i &= \frac12 \Omega^i_{bc} \alpha^b\wedge\alpha^c\\
	\d \alpha^a + \wb{\omega}{\alpha}^a &= \frac12 \Omega^a_{bc} \alpha^b\wedge\alpha^c
\end{align}
\end{subequations}

From these equations it is possible to build a Lie algebra action on the manifold. In fact, this mechanism is well known from the group manifold approach to supergravity~\cite{GeoSUGRA,SUGRAI}. However more hypotheses (of a global nature) are needed for the action to integrate to a Lie group action, and for the fibration over the orbit space to form a principal bundle fibration. For this reason, we will call \enquote{generalised frame bundle} a manifold equipped with such a coframe. We do not discuss the comparison with standard frame bundles: it will be the subject of an upcoming publication.

In~\cite{LFB} field equations on the base manifold are obtained in the following way: one first uses the Euler-Lagrange equations associated to the variation of an extra field that plays the role of Lagrange multipliers. They allow to construct the frame bundle structure on the $10$-manifold. Then in order to deal with the remaining Euler-Lagrange equations, associated to variations of the fields $\omega$ and $\alpha$, a change of fibre coordinates is applied that amounts to choosing a local frame on the four dimensional base manifold. This allows converting equivariant quantities on the frame bundle to \emph{invariant quantities}. In these coordinates the Euler-Lagrange equations decompose into two terms: one is manifestly an exact divergence and the other one is invariant under the group action. Integration along fibres allows concluding that the invariant part of the Euler-Lagrange equations has to vanish, which turns out to be equivalent to the Einstein field equations on spacetime.

Once a part of Euler-Lagrange equations is used to obtain the (local) frame bundle structure, we approach the problem of the field equations from a different perspective. Instead of considering the Euler-Lagrange equations corresponding to field variations associated with a local trivialisation and specific fibre coordinates, it is enough to consider variations of the fields which are \emph{equivariant} with respect to the action of the structure group. The calculations are in a sense equivalent but our perspective makes it clear why and how the Euler-Lagrange equations split into two terms that respectively give an exact term and the usual field equations on the base manifold. With this insight it is simple to extend the theory to include spinor fields. One peculiarity of the spinor theory thus constructed is that instead of starting from a spacetime and a tetrad which connects the spacetime to a reference bundle in which spinor fields live, we start from the tentative frame bundle, on which the spinor fields are maps to a fixed vector space, and from there construct back the spacetime.

The paper is organised as follows. 
In Section~\ref{secno:EC} we introduce the physical Lagrangians for the \emph{Einstein-Cartan theory} of gravitation with torsion, and its minimal coupling with a \emph{Dirac spinor}. Section~\ref{annRic} starts with a quick presentation of the structure of \emph{frame bundles} and connections thereon. We then motivate and define the structure of \emph{generalised frame bundles}. Our Lagrangian field theory is finally introduced in Section~\ref{secno:Lag}, in which we also derive the Euler-Lagrange equations. The theory is defined on a structure-less $10$-manifold which is meant to acquire a structure of generalised frame bundle. We show in Section~\ref{secno:EuclFE} that when the generalised frame bundle structure is that of a standard frame bundle, the Euler-Lagrange equations imply the Einstein-Cartan-Dirac field equations on the underlying $4$-dimensional spacetime. The derivation relies crucially on a \emph{cancellation mechanism}, which allows decoupling out an exact term from the Euler-Lagrange equations.
\subsection{Notations and conventions}\label{secno:notconv}

We will make free use of the Einstein convention:
\[ A_iB^i := \sum_i A_i B^i \]
as well as of the so-called \enquote{musical} isomorphisms raising and lowering indices through a metric that will generally be implicit (as long as there is no ambiguity). The convention for the order of the indices will be to keep the relative order of upper indices and lower indices, and place upper indices \emph{before lower indices}. For example given a tensor $T^\pi{}_{\mu\nu}$ the corresponding totally covariant and totally contravariant tensors are
\begin{align*}
	T_{\tau\mu\nu} &:= g_{\tau\pi}T^\pi{}_{\mu\nu}\\
	T^{\pi\tau\rho} &:= g^{\tau\mu}g^{\rho\nu}T^\pi{}_{\mu\nu}
\end{align*}
given a metric $g$ on the space in which the indices $\tau,\mu,\nu$ live. The metric and the inverse metric will be both written with the same symbol, but we will generally use the Kronecker delta $\delta^\mu_\nu = g^{\mu\tau}g_{\tau\nu}$ for the identity which is the corresponding endomorphism of both covariant and contravariant vectors.

We will write $\met_{ab}$ for the Minkowski metric on the Minkowski space $\m$, used at the same time as an abelian Lie group and as an abelian Lie algebra. Our convention for the Lorentzian signature is $(+---)$ and for the Clifford algebras $u\cdot v + v\cdot u = -2\sprod u v$. We will be working with the connected \emph{proper orthochronous Lorentz group} $\SOpL$ which we will just call \emph{Lorentz group}. Its Lie algebra is $\lor$.
The \emph{Poincaré group} is isomorphic to the semi-direct product $\Lor \ltimes \m$ 
and its associated Lie algebra is $\lor\ltimes \m$

The differentiable manifolds we consider are always \emph{finite dimensional second countable Hausdorff}.

We will be using for vector-valued differential forms a duality notation convention that is described in detail in Appendix~\ref{anndual}.



\section{Spacetime with torsion: the Einstein-Cartan theory}\label{secno:EC}
We start with a very succinct presentation of the Einstein-Cartan theory of gravitation and its coupling to Dirac spinor fields. 

\paragraph{Einstein's General Relativity}
In Einstein's original theory, gravitation is modelled by \emph{the curvature} of a Lorentzian metric a the differentiable manifold modelling the spacetime. More accurately, the Newtonian gravitational field is replaced by the torsion-free Levi-Civita connection associated to the metric, which defines the geodesic equation governing inertial trajectories. Newton's formula for the interaction between massive bodies and Poisson's equation for the gravitational potential are replaced by Einstein's field equation, which is a second order differential equation, and the geodesic equation. 

Einstein's general relativity can be formulated as a Lagrangian field theory as follows: spacetime is modelled by a $4$-dimensional manifold $\ET$ with a dynamical metric $g$ of Lorentzian signature. The (adimensional) Lagrangian is the \emph{Hilbert-Einstein Lagrangian}:
\begin{equation}\label{eqno:EHLag}
	\Lag_{\text{HE}}[g] = \Scal^g \vol^g
\end{equation}
with $\vol$ the Lorentzian volume density associated with $g$ and $\Scal^g$ the scalar curvature.

\paragraph{Palatini's Lagrangian}

It was then realised that a first order formulation is possible if one was to consider as possible field a couple gathering the metric $g$ itself and an a priori independent affine connection $\nabla$. This is called the \emph{Palatini formalism} and exists in several variations, requiring the connection to be metric or not (see \cite{MetricAffine1,MetricAffine2,Palatini}). In the case the connection is neither assumed to be metric nor torsion-free an extra gauge freedom appears, the so-called projective symmetry \cite{MetricAffine1,Palatini}. The \emph{Palatini Lagrangian} is formally similar to~\eqref{eqno:EHLag}:
\begin{equation}\label{eqno:PalLag}
	\Lag_{\text{Palatini}}[g, \nabla]
		 = \tr_g \lp \Ric^\nabla \rp \vol^g
\end{equation}

\paragraph{Tetradic formulation}

Within the Palatini formalism, it is possible to make explicit a Lorentzian gauge symmetry: this is the so-called \emph{tetradic Palatini formalism}, in which the metric field is replaced by a field of linear frames, or equivalently (in our case) \emph{coframes}. The linear frame is called a \emph{tetrad}, or \emph{vierbein} (\emph{vielbein} in general dimension). The \emph{Einstein-Cartan theory} (more precisely Einstein-Cartan-Sciama-Kibble) 
\cite{GRSpinTors}
is the case in which the connection is assumed to be \emph{metric} but is allowed to have torsion.

The Lagrangian has a convenient expression if instead of working with the affine connection on spacetime we use the image connection $\nabla$ on $\setR^{1, 3}\times \ET\to \ET$ under the tetrad $e^a : T\ET \to \setR^{1, 3}\times \ET$. Since $\setR^{1,3}\times \ET$ is a trivial vector bundle, $\nabla$ has a global connection $1$-form $\omega$ 
and the Lagrangian takes the following form:
\begin{equation}\label{eqno:TetradLag}
	\Lag_{\text{Tetrad}}[e, \omega]
		= \eta^{ac}\lp \d \omega + \fwb\omega\omega \rp_{a}^{b} \wedge e^{(2)}_{bc}
\end{equation}
with 
\begin{itemize}
	\item $\eta^{ac}$ the inverse Lorentzian metric on $\setR^{1,3}$,
	\item $\d\omega + \fwb \omega\omega := \Curv \in \Omega^2(\ET, \End(\setR^{1,3}))$ the curvature endomorphism of $\nabla$,
	\item $e^{(2)}_{ab}$ the dual $2$-forms as defined in Appendix~\ref{anndual}.
\end{itemize}
This is the form of the Lagrangian we will be using.

\paragraph{The Einstein-Cartan field equations}

The field equations are the Euler-Lagrange equations associated to the Lagrangian~\eqref{eqno:TetradLag}. They take the following form:
\begin{subequations}
\begin{equation}
	\Curv^b_a \wedge e^{(1)}_{dbc}
		= 0
\end{equation}
\begin{equation}
	T^d\wedge e^{(3)}_{bcd} = 0
\end{equation}
\end{subequations}
with $T^d = \d e^d + [\omega, e]^d$ the torsion of the affine connection, transported back to a $\aux$-valued $2$-form.

These equations are equivalent to the more usual equation
\begin{subequations}
\begin{equation}
	2\Ric^\omega - \tr_g \Ric^\omega \eta = 0
\end{equation}
along with a similar equation on torsion (after multiplication by $1/2$)
\begin{equation}
	T + 2\tr(T)\wedge \id = 0
\end{equation}
\end{subequations}

\paragraph{Einstein-Cartan-Dirac theory}

The Einstein-Cartan theory can be coupled to dynamical Dirac spinors. We need to introduce the relevant structures:
\begin{enumerate}
	\item $\Sp$ is an irreducible complex module over the Clifford algebra $\Cl_{1,3}$,
	\item $\so_{1,3}\hookrightarrow \Cl_{1,3}$ acts on $\Sp$ by endomorphisms $\sigma_a^b$ which preserve (infinitesimally) a hermitian product $\sprod{\psi_1}{\psi_2} =: \bar{\psi_1}\psi_2$,
	\item $\gamma^a : {\setR^{1,3}}^* \to \End(\Sp)$ represent the action of dual vectors on the spinor module,
	\item The $\End(\Sp)$-valued $3$-form $\gamma^{(3)} = \gamma^a e^{(3)}_a$.
\end{enumerate}

Let $\psi$ be a $\Sp$-valued field on $\ET$. Its covariant derivative takes the form
\[
	\d \psi + \omega^a_b \sigma_a^b \psi
\]
The Einstein-Cartan-Dirac Lagrangian is the following:
\footnote{A careful physical coupling between the two terms would use coupling constants.}
\begin{equation}\begin{multlined}\label{eqno:ECDLag}
	\Lag_{\text{ECD}}[e, \omega, \psi]
		=
			\eta^{ac}\Curv_{a}^{b} \wedge e^{(2)}_{bc} 
			- \frac12 \lp
				\bp \gamma^{(3)} \wedge (\d \psi + \omega^a_b \sigma_a^b \psi)
				+
				(\d \bp + \omega^a_b \bar\sigma_a^b \bp) \wedge \gamma^{(3)} \psi
			\rp
			+ m^2 \bp \psi e^{(4)}
\end{multlined}
\end{equation}
with $\bar\sigma^b_a$ the transpose of $\sigma^b_a$.

\paragraph{Einstein-Cartan-Dirac equations}
\com{Einstein-Cartan-Dirac-Sciama-Kibble ?}

The field equations are the Euler-Lagrange equations corresponding to the Lagrangian~\eqref{eqno:ECDLag}. Using a notation $\nabla\psi = \d \psi + \omega^a_b \sigma^b_a \psi$ for the covariant derivative, they take the following form:
\begin{subequations}\label{eqnos:ECD}
\begin{equation}
	\eta^{ac}\Curv^b_a \wedge e^{(1)}_{bcd}
		= \frac12\lp
			\bp \gamma^c e^{(2)}_{cd}\wedge \nabla \psi
			- \nabla \bp \wedge \gamma^c e^{(2)}_{cd} \psi \rp
			+ m\bp \psi e^{(3)}_d
\end{equation}
\begin{equation}
	T^d\wedge e^{(3)}_{bcd} = -\frac1{16} \bp \{ [\gamma_b, \gamma_c], \gtrois \} \psi
\end{equation}
\begin{equation}
	- \gtrois \wedge \nabla\psi + \frac12\gamma^a \psi T^b \wedge \adeux{ab} - m\psi = 0
\end{equation}

\end{subequations}

Our manipulations will be clearer if we use \enquote{generalised tetrads} which are isomorphisms between $T\ET$ and a fixed reference vector bundle $\aux$. This is also convenient from a geometric perspective since their global existence does not require spacetime to be parallelisable. 
\footnote{Note however that the existence of a (real) $\Spin$ structure on a noncompact Lorentzian manifold requires it to be orthonormally parallelisable~\cite{GerochSpinorI}.}
The Lorentzian product $\eta$ is thus replaced by a product $\eta$ of Lorentzian signature on the vector bundle $\aux$. The $1$-form $\omega$ is replaced by a non-trivialised connection on $\aux$.

The theory we want to study is a version of the Einstein-Cartan-Dirac theory formulated on a \enquote{generalised frame bundle}. This structure is introduced in the following section.

\section{Frame bundle, Ricci curvature and torsion}\label{annRic}
In this section we recall the framework of frame bundles and introduce that of the so-called \enquote{Generalised Frame Bundles}.
%
%

\subsection{Frame bundles}\label{secno:FrameB}
\subsubsection*{Structure of frame bundles}

We start with a brief reminder on the structure of the bundle of orthonormal frames. A general reference is~\cite{KobaNomi1}. Everything done in this section applies in a straightforward manner to any metric signature and to unoriented frame, time-oriented frame and space-oriented bundles as well as to spin (and pin) frame bundles.

Let us start with an $n$-manifold $M$, provided with a metric $g$ of signature $(p,q)$ as well as an orientation. 
The bundle of positive orthonormal frames $\SOp(M)\overset{\pi}{\to} M$ has as underlying set
\[
			\bigcup_{m\in M} 
					\bigg\{
						\Iso\Big( (\setR^{p+q}, \eta_{p,q})
							, (T_m M, g|_m) \Big)
					\bigg\}
\]
with $\eta_{p,q}$ the metric of signature $(p,q)$ on $\setR^{p+q}$.
It has the structure of a $\SOp_{p,q}$-principal bundle.

On the other hand, there is the bundle of \emph{linear frames} $\GL(M)$, which is defined similarly but using \emph{vector space isomorphisms}. It is equipped with a \enquote{canonical soldering form}~\cite{KobaNomi1, MichorTopics} defined according to the following diagram:
\begin{center}
\begin{tikzcd}
	T\GL(M)
		\ar[d]
		\ar[dr,"\alpha"] &
	\\
	T\GL(M)/V\GL(M)\simeq \pi^*TM
		\ar[r] &
	\setR^n
\end{tikzcd}
\end{center}
The canonical soldering form restrics to a \enquote{soldering form} on $\SOp(M)$, which we now define:
\begin{proposition}\label{propno:soldering}
	Let $G$ be a Lie group with an action on $\setR^n$. Let $P\xrightarrow{\pi} M$ be a $G$-principal bundle.
	There is a natural equivalence between the following data:
	\begin{enumerate}
	\item A $G$-equivariant bundle map $P\to \GL(M)$ to the bundle of linear frames of $TM$,
	\item A $G$-equivariant, horizontal $\setR^n$-valued $1$-form $\alpha$ on $P$ which is onto $\setR^n$ at each point.
	\item A vector bundle isomorphism $TM \isom P\times_G \setR^n$.
	\end{enumerate}
\end{proposition}
\begin{proof}
	The mapping $1. \implies 2.$ is simply pulling back the canonical soldering form of $\GL(M)$. It gives a $G$-equivariant form which is also horizontal (since the bundle map maps vertical vectors to vertical vectors).
	The inverse mapping is constructed as follows: since $\alpha : TP\to \setR^n$ is horizontal, it factors to a map $\pi^*TM \to \setR^n$ which has to be an isomorphism over each point of $P$. This gives the following commutative diagram of vector bundle above $P$
	\[\begin{tikzcd}
		TP	
			\ar[r, "\alpha"]
			\ar[d, "\d \pi"]
		&
		P\times \setR^n
		\\
		\pi^*TM 
			\ar[ru, dashed, "\sim"' sloped]
		&
	\end{tikzcd}\]
	with the factorised map which is $G$-equivariant. 
	It thus defines above each point $p$ of $P$ a vector space isomorphism $\setR^n \isom T_{\pi(p)} M$, and this in a smooth and equivariant fashion: this is exactly an equivariant bundle map $P\to \GL(M)$.
	
	The mapping $2. \implies 3.$ is done by factoring the $G$-equivariant map $\pi^* TM \isom P\times \setR^n$ under the quotient by $G$: this gives a vector bundle isomorphism above $M$ : $\pi^*TM/G \equiv TM \isom P\times_G \setR^n$. The inverse mapping is just pulling back under the bundle fibration $P \xrightarrow{\pi} M$ the isomorphism $TM \isom P\times_G \setR^n$ and composing with the natural trivialisation $\pi*\lp P\times_G \setR^n \rp \equiv P\times \setR^n$.
\end{proof}

Datum 1. is a reduction of structure group of $\GL(M)$ along $G\to \GL_n$; we call a form satisfying the properties of 2. a \emph{soldering form}.

\subsubsection*{Spin structures}
A \emph{$\Spinp_{p,q}$}-structure is given by a \enquote{lifting} of the so-called \emph{structure group} $\SOp_{p,q}$ to $\Spinp_{p,q}$. In other words, given a metric and a space-and-time orientation, it is defined by a $\Spinp_{p,q}$-principal bundle $P \to M$ with an equivariant bundle map $P\to \SOp(M)$. 
\footnote{Note that a similar lifting of the positive \emph{linear} frame bundle to a principal bundle with the connected double cover of $\GL_n^+$ as structure group induces a $\Spinp_{p,q}$-structure for every metric and space-and-time orientation.}
According to Proposition~\ref{propno:soldering}, the mapping $P \to \SOp(M)$ is equivalent to the data of the solder form pulled back to $P$: the $\Spinp_{p,q}$-structure corresponds to a principal $\Spinp_{p,q}$-bundle equipped with an nondegenerate horizontal $\Rpq$-valued $1$-form which is equivariant for the action $\Spinp_{p,q}\to \SOp_{p,q}$.

\subsubsection*{Connection forms}
A (principal) connection 1-form on $\SOp(M)$ (hereafter \emph{connection form}) 
is given by an $\sopq$-equivariant $\sopq$-valued 1-form
\footnote{
	This is different, although related, from the $\omega$ used in Section~\ref{secno:EC} which is a $1$-form on $\ET$ and not $\tot$.
}
$\omega$ on $\SOp(p,q)$, which is normalized for the action of $\sopq$. Namely, writing $\bh\in\Gamma(T\SOp(M))$ the action on $M$ of an element $\h\in\sopq$, the normalisation condition is
\begin{equation} \omega \lp \bh \rp = \h 	\label{eqno:omnormcond}\end{equation}
A connection form defines an Ehresmann connection (horizontal distribution) given by its kernel. The equivariance of the form ensures that the horizontal distribution is equivariant. The combined data of 
\begin{equation}
	\omega\oplus\alpha\in\Omega^1(\SOp(M),\sopq\ltimes\Rpq)
\end{equation}
is called a \emph{Cartan connection} form, or \emph{affine connection}.

\subsubsection*{Tensorial forms}
%

Similarly to the identification $\pi^*TM \isom \SOp(M)\times \Rpq$, the solder form induces an equivariant trivialisation of the pullbacks of tensor bundles to $\SOp(M)$: 
\[
	\pi^*\lp TM^{\otimes k}\otimes T^*M^{\otimes l} \rp \to \Rpq^{\otimes k}\otimes \Rpq^{*\otimes l}
\]
As a consequence, any tensor-valued differential form on $M$ pulls back to $\SOp(M)$ to a differential form with values in a trivialised bundle, which is a \emph{horizontal} (contracts to $0$ with any vertical vector) and \emph{equivariant} form. Such forms on $\SOp(M)$ are called \emph{basic} or \emph{tensorial} 
\footnote{There seems to be an ambiguity regarding this terminology. \cite{KobaNomi1} uses \enquote{tensorial forms} and calls \enquote{pseudotensorial} equivariant forms without any horizontality requirement; \cite{MichorTopics, SGAM} uses \enquote{basic forms} and calls \enquote{semibasic} horizontal forms without equivariance requirement.}
and are in bijection with forms of the corresponding tensorial type on $M$. 
We write 
\[
	\Omhor^\bullet \lp \SOp(M),\, \Rpq^{\otimes k}\otimes \Rpq^{*\otimes l} \rp^{\SOpL}
\]
for the space of $\Rpq^{\otimes k}\otimes \Rpq^{*\otimes l}$-valued tensorial forms. There is the following more general result:
\begin{proposition}[\cite{KobaNomi1}, Example II.5.2]
	Let $G$ be a Lie group acting on a vector space $\Sp$. 
	Let $P\to M$ be a $G$-principal bundle: there is an associated vector bundle $P[\Sp] := P\times_G \Sp$ of fibre $\Sp$.
	There is a natural bijection of graded vector spaces
	\[
		\Omega^k(M, P[\Sp])
		\overset{\sim}{\longleftrightarrow}
		\Omhor^k(P, \Sp)^G
	\]
\end{proposition}
\com{Il faut une référence...}

\subsubsection*{The equivariance criterion}

Equivariance under the connected group $\SOpL$ is equivalent to equivariance under $\soL$, which can be written as follows for a $\Rpq^{\otimes k}\otimes \Rpq^{*\otimes l}$-valued form $\psi$ on $\SOp(M)$: 
\begin{equation}\label{eqno:equivsigma}
	\forall \xi\in \sopq, \quad
		(\d i_\bh + i_\bh \d) \psi + \h\cdot \psi = 0
\end{equation}
Since $\psi$ is horizontal, Equation~\eqref{eqno:equivsigma} is equivalent to
\[ (i_\bh \d + \h\cdot)\psi = 0 \]
Since $\omega$ vanishes on the horizontal directions, is normalized~(Equation \eqref{eqno:omnormcond}) and since $\alpha^a$ span the horizontal forms, the infinitesimal equivariance  can be equivalently written as
\begin{equation}\label{eqno:equivFrameB}
	\d\psi + \omega\cdot \psi = S_{\mathcal I}\alpha^{\mathcal I}
\end{equation}
where ${\mathcal I}$ are multi-indices, $\alpha^{\mathcal I}$ form a basis of the horizontal forms and $S_{\mathcal I}$ are (variable) coefficients determined by $\sigma$.

\subsubsection*{Covariant derivative}

Given a $\Rpq^{\otimes k}\otimes \Rpq^{*\otimes l}$-valued tensorial $k$-form $\psi$ on $\SOp(M)$, it has a covariant derivative which is defined as
\[
	\dom \psi := \d \psi + \omega \cdot \psi \in \Omega^{k+1} \lp \SOp(M), \Rpq^{\otimes k}\otimes \Rpq^{*\otimes l} \rp
\]
with $\omega\cdot \psi$ implying a wedge product on the differential forms. The form $\dom \psi$ is a tensorial $(k+1)$-form of the same type and therefore corresponds to a tensor-valued $(k+1)$-form on $M$.

\subsubsection*{Curvature and Torsion forms}

To a principal connection $\omega$ on $\SOp(M)$ are associated its torsion 2-form $\Theta \in \Omega^2(\SOp(M), \Rpq)$ and its curvature 2-form $\Omega\in \Omega^2(\SOp(M), \sopq)$ defined by
\begin{align}
	\Theta &:= \d\alpha + \omega \cdot \alpha\\
	\Omega &:= \d\omega + \fwb\omega\omega
\end{align}
with 
\[
	(\omega\cdot \alpha)(X,Y) = \underbrace{\omega(X)}_{\sopq}\cdot \underbrace{\alpha(Y)}_{\Rpq}
		- \underbrace{\omega(Y)}_{\sopq}\cdot \underbrace{\alpha(X)}_{\Rpq}
\]
and
\[
	\fwb\omega\omega (X,Y) =
		\left[\omega(X), \omega(Y)
		\right]
\]
They both are horizontal and equivariant and are as such associated to fields on $M$.

We will be using throughout the following indices convention:
\begin{itemize}
\item indices $a,b \dots$ used for $\Rpq$,
\item indices $i,j \dots$ used for $\sopq$,
\item indices $A,B \dots$ used for $\sopq\ltimes \Rpq$,
\item indices $\alpha, \beta \dots$ used for the spinor module $\Sp$.
\end{itemize}

\subsection{Generalised Frame Bundles}

In this section we define the structure of \enquote{generalised frame bundles}, a characterisation of the \emph{local structure} of the total spaces of frame bundles \emph{equipped with a principal connection}. In particular, it does not involve a group action nor a locally trivial fibration. The definition will rely on the following proposition.

\begin{proposition}\label{propno:GFB}
	Let $p,q\in \setN$ and let $\tot$ be a smooth manifold. Assume it is provided with a $1$-form $\varpi = \omega\oplus \alpha \in \Omega^1 \lp \tot, \sopq\ltimes\Rpq \rp$.
	Assume furthermore that $\varpi$ defines a coframe, namely a vector bundle isomorphism $T\tot \isom \tot \times \sopq\ltimes\Rpq$.
	
	Then the following assertions are equivalent:
	\begin{enumerate}
	\item There exist (variable) coefficients $\Omega^i_{bc}, \Theta^a_{bc}$ such that
		\begin{subequations}\label{eqnos:GFBao}
		\begin{align}
			\d \alpha^a + \omega\cdot \alpha &= \frac12 \Theta^a_{bc} \alpha^b\wedge \alpha^c\\
			\d \omega^i + \fwb\omega\omega^i &= \frac12 \Omega^i_{bc} \alpha^b\wedge \alpha^c
		\end{align}
		\end{subequations},
		
	\item The map $\xi\in \sopq \mapsto \bar\xi = \varpi^{-1}(\xi, 0) \in \Gamma(T\tot)$ defines an action of the Lie algebra $\sopq$ on $\tot$ for which $\varpi$ is equivariant.
	\end{enumerate}
\end{proposition}

\begin{remark}
	\begin{enumerate}
	\item Equations~\eqref{eqnos:GFBao} can be gathered into a single equation
		\begin{equation}\label{eqno:GFBvarpi}
			\d \varpi^A + \fwb\varpi\varpi^A = \frac12\Omega^A_{bc} \alpha^b\wedge \alpha^c
		\end{equation}
		using the notation $\Omega^A$ gathering both $\Theta^a$ and $\Omega^i$.
	\item Such a structure was already defined in~\cite{DiffGeoCartan} but for different geometrical purposes.
	They write the equivalent condition
		\[
			\forall \xi\in \sopq, \zeta\in \sopq\ltimes \Rpq, \quad
				[\bar\xi, \bar\zeta] = \overline{[\xi, \zeta]}
		\]
	with $\zeta\in \sopq\ltimes\Rpq \mapsto \bar\zeta:= \varpi^{-1}(\zeta) \in \Gamma(T\tot)$,

	\item The action of the Lie algebra $\sopq$ is necessarily free.
	
	\item Since there is no distinction between the Lie algebras $\sopq$ and $\spin_{p,q}$ the local structure of a generalised frame bundle cannot make the distinction (it requires global considerations).
	\end{enumerate}
\end{remark}
\begin{proof}
	We are going to prove the following: $\forall \xi\sopq, \ \Lie_{\bh} \varpi +\xi \cdot \varpi = 0$ if and only if $\varpi$ satisfies Equation~\eqref{eqno:GFBvarpi}. The former implies that 
	\[
		\forall\xi\in \sopq, \zeta\in \sopq\ltimes \Rpq, 
		[\bar\xi, \bar\zeta] = \overline{[\xi, \zeta]}
	\]
	and therefore that $\bar\xi$ define an action of $\sopq$ on $\tot$.
	
	Let $\bar\xi\in \sopq$. Then we compute
	\[
		\Lie_{\bar\xi} \varpi
			= \lp i_{\bar\xi} \d + \d i_{\bar\xi} \rp \varpi
			= i_{\bar\xi} \d \varpi + \d \xi
			= i_{\bar\xi} \d \varpi
	\]
	and
	\[
		i_{\bar\xi}\wb\varpi\varpi = 2[\xi, \varpi]
	\]
	so that
	\[
		i_{\bar\xi}\lp \d \varpi + \fwb\varpi \varpi \rp = \Lie_{\bh}\varpi + \xi\cdot \varpi
	\]
	
	To conclude, it is enough to remark that $i_{\bar\xi}\lp \d \varpi + \fwb\varpi \varpi \rp$ vanishes for all $\xi\in \sopq$ if and only if it is purely horizontal, namely it decomposes over the horizontal $2$-forms $\alpha^a\wedge \alpha^b$.
\end{proof}

\begin{definition}[Generalised frame bundle]
	Let $\tot$ be a smooth manifold.
	The structure of a \emph{generalised frame bundle} of signature $(p,q)$ on $\tot$ is the datum of a $\sopq\ltimes \Rpq$-valued coframe $\varpi^A$ satisfying Equation~\eqref{eqno:GFBvarpi}.
\end{definition}
\begin{remark}
	The existence of a $\sopq\ltimes \Rpq$-valued coframe on $\tot$ requires it to be of dimension $\frac{n(n+1)}{2}$ with $n=p+q$, and furthermore to be parallelisable.
\end{remark}
It is straightforward from the definition that a frame bundle equipped with its soldering form $\alpha^a$ and a connection form $\omega^i$ has a structure of generalised frame bundle.

A generalised frame bundle has an action of the Lie algebra $\sopq$ and a decomposition of its tangent bundle into horizontal and vertical parts:
	\[
		T\tot = \underbrace{{V\tot}}_{
					\substack{\alpha=0\\
							\simeq \so_n}}
				\oplus \underbrace{{H\tot}}_{
					\substack{\omega=0\\
							\simeq \setR^n}}
	\]

It has enough structure to define the following objects, which coincide with the standard notions in the case of a usual frame bundle:
\begin{itemize}
	\item \emph{Basic vector fields} which are $\sopq$-equivariant maps $\tot \to \Rpq$,
	\item \emph{Basic tensor fields} or more generally \emph{basic sections of associated vector bundles} which are $\sopq$-equivariant maps $\tot \to M$ into a linear representation $M$ of $\sopq$. In particular for spinor modules there are basic spinor fields.
	\item \emph{Covariant derivative} of basic sections (and vector bundle-valued forms) which are computed using
	\[
		\dom := \d + \omega\wedge
	\]	
	\item \emph{Curvature} and \emph{torsion} computed from $\varpi = \omega\oplus \alpha$:
	\begin{align*}
		\Omega^A &:= \d\varpi^A + \fwb\varpi\varpi^A
	\end{align*}
	The $a$ components correspond to the torsion while the $i$ components correspond to the curvature. It is by hypothesis a horizontal $\sopq\ltimes \Rpq$-valued $2$-form, which is always $\sopq$-equivariant. In particular, it is possible to formulate Einstein's equation on a generalised frame bundle.
\end{itemize}

The notation $\dom$ will be freely used for arbitrary forms with values in a representation of $\sopq$, not necessarily horizontal or equivariant.

Equivariant quantities on $\tot$ can be characterised using the following lemma, which is the generalised frame bundle version of the equivariance criterion of Section~\ref{secno:FrameB}:
\begin{lemma}\label{lmno:equivdom}
	Let $M$ be a representation of $\sopq$ and $(\tot, \varpi = \omega\oplus \alpha)$ be a $(\sopq\ltimes\Rpq)$-valued generalised frame bundle.
	Let $\Psi$ be a $M$-valued \emph{horizontal} $k$-form on $\tot$.
	Then $\Psi$ is $\sopq$-equivariant \emph{if and only if} $\d\Psi + \omega \wedge \Psi$ is horizontal, in which case it is necessarily equivariant.
\end{lemma}
\begin{proof}
	The argument is almost the same as for the proof of Proposition~\ref{propno:GFB}. Let $\xi\in \sopq$ and $\bar\xi$ be the associated vector field on $\tot$. Then under the assumption mentioned in the statement, 
	\[
		\Lie_{\bar\xi}\Psi
			= \lp i_{\bar\xi}\d + \d i_{\bar\xi} \rp \Psi
			= i_{\bar\xi} \d \Psi + 0
	\]
	and 
	\[
		i_{\bar\xi} \lp \omega\wedge \Psi \rp
			= \omega \lp \bar\xi \rp \Psi - \omega\wedge i_{\bar\xi}\Psi
			= \xi\cdot \Psi + 0
	\]
	
	One thus concludes that $\Lie_{\bar\xi}\Psi + \xi\cdot \Psi$ vanishes for all $\xi\in \sopq$ if and only if $\d\Psi + \omega\wedge \Psi$ is purely horizontal.
	
	If this is the case, then since $\Lie_{\bar\xi}$ commutes with $\d$, $\d\Psi$ is equivariant and so does $\omega\wedge\Psi$. 
	As a consequence, $\d\Psi + \omega\wedge \Psi$ is equivariant.
\end{proof}
\begin{remark}
	The lemma and its proof hold even when the coframe is not a generalised frame bundle, defining equivariance in this case by the characterisation used in the proof.
\end{remark}

The question of relating and comparing generalised frame bundles to usual frame bundles is discussed in~\cite{PDMPhD} and will be the subject of an upcoming publication~\cite{CartInt}.
We now have all the tools in hands to introduce our Lagrangian theory and study it.

\section{Einstein-Cartan-Dirac theory on a generalised frame bundle}\label{secno:Lag}
We want to build a Lagrangian field theory on a $10$-dimensional manifold $\tot$ such that its solution endow $\tot$ with the structure of generalised frame bundle which satisfies (a generalised frame bundle version of) the ECD Equation~\eqref{eqnos:ECD}. The model we construct does not satisfy the \enquote{natural} generalised frame bundle version of the equations, however we will show in Section~\ref{secno:EuclFE} that (part of) the Euler-Lagrange equation imply as expected the ECD equations \emph{when the generalised frame bundle is a classical frame bundle}.

\subsection{The Lagrangian}\label{secno:Lagrangian}

Let us recall the ECD Lagrangian using trivial tetrads on a $4$-manifold $\ET$:
\[
	\Lag_{\text{ECD}}[e, \omega, \psi]
			=
		\eta^{ac}\Omega_{a}^{b} \wedge e^{(2)}_{bc} \\
		- \frac12 \lp
			\bp \gamma^{(3)} \wedge (\d \psi + \omega^a_b \sigma_a^b \psi)
			+
			(\d \bp + \omega^a_b \bar\sigma_a^b \bp) \wedge \gamma^{(3)} \psi
		\rp
\]

\subsubsection*{Tetrads from a frame bundle perspective}

The tetrad $e$ is a vector bundle isomorphism $T\ET \isom \setR^{1,3}\times \ET$. We propose to replace $\setR^{1,3}\times \ET$ with an arbitrary vector bundle $\aux$ equipped with a Lorentzian product $\eta_\aux$ and a space-and-time orientation. The corresponding \enquote{generalised}
\footnote{This is unrelated to the notion of generalised frame bundles.}
tetrad is a vector bundle isomorphism $T\ET \isom \aux$.
Using the following chain of equivalences
\begin{enumerate}
\item Vector bundle isomorphism $T\ET \isom \aux$
	\\ $\Longleftrightarrow$
\item $\GL_n(\setR)$-principal bundle isomorphism $\GL(\ET) \overset{\sim}{\longleftrightarrow} \GL(V)$ with $\GL(V)$ the bundle of linear frames of $V$
	\\ $\Longleftrightarrow$
\item Equivariant embedding $ \SOp(V) \hookrightarrow_{\SOp_{1,3}} \GL(T\ET)$
	\\ $\Longleftrightarrow$
\item Soldering form on $\SOp(V)$.
\end{enumerate}
the tetrad is seen to be an equivalent datum to a soldering form on $\SOp(V)$.


Since we want to handle spinors, $\aux$ will be assumed to be furthermore equipped with a spin structure and we will consider a \enquote{spin frame bundle} $\Spinp(\aux)$ with structure group $\Spinp_{1,3}$ instead of $\SOp(\aux)$. We will use indices $I, J, \dots$ for tensorial components in $\aux$ (or $\aux^*$).

\subsubsection*{The Einstein-Cartan term}

Focusing first on the term $\eta^{ac}\Curv_{a}^{b} \wedge e^{(2)}_{bc}$ replaced on $\aux$ by $\eta_\aux^{IK}\Curv_I^J \wedge e^{(2)}_{JK}$, let us express the pullback to $\Spinp(\aux)$.

The term $\eta_\aux^{IK}\Curv^J_I$ is a $2$-form on $\tot$ with value in $\ExT^2 \aux$. The associated basic $2$-form on $\Spinp(\aux)$ is 
\[
	\eta^{ac}\Omega^i \rho_{i}{}^b_a \in \Omhor^2 \lp \Spinp(\aux), \ExT^2 \setR^{1,3} \rp
\]
with $\rho_{i}{}^b_a : \so_{1,3}\to \End(\setR^{1,3})$ the standard representation and $\Omega^i = \d \omega^i + \fwb\omega\omega^i$.

The term $e^{(2)}_{JK}$ can be understood as a $\ExT^2 \aux^*$-valued $2$-form, and as such is associated to
\[
	\alpha^{(2)}_{bc} \in \Omhor^2 \lp \Spinp(\aux), \ExT^2 {\setR^{1,3}}^* \rp
\]
As a consequence the term $\eta_\aux^{IK}\Omega_{I}^{J} \wedge e^{(2)}_{JK}$ lifts to
\begin{equation}
	\eta^{ac}\Omega^i \rho_{i}{}^b_a \wedge \alpha^{(2)}_{bc}
\end{equation}

\subsubsection*{The Dirac term}

We now want to express the pullback of a term 
\[
	\frac12 \lp
			\bp \gamma^{(3)} \wedge \nabla \psi
			+
			\nabla \bar\psi \wedge \gamma^{(3)} \psi
		\rp
		- m\bp\psi e^{(4)}
\]
to $\Spinp(\aux)$.

The spinor field $\psi$ corresponds to an equivariant map $\Psi : \Spinp(\aux)\to \Sp$ and its covariant derivative corresponds to 
\[
	\dom \Psi = \d \Psi + \omega\cdot \Psi \in \Omhor^1(\Spinp(\aux), \Sp)
\]
The term is thus pulled back to
\begin{equation}
	\frac12 \lp
		\bar\Psi \gamma^{(3)} \wedge \dom \Psi
		+
		\dom\Psi \wedge \gamma^{(3)} \Psi
	\rp
	- m \bP \Psi \aquatre
\end{equation}

\subsubsection*{Lagrange multipliers}

We want to define a field theory on a $10$-manifold: we need for this a Lagrangian $10$-form.
Let us turn the pulled back $4$-forms into $10$-forms:
\begin{multline*}
	\lp
		\eta^{ac}\Omega^i \rho_{i}{}^b_a \wedge \alpha^{(2)}_{bc}
		- 	\frac12 \lp
				\bar\Psi \gamma^{(3)} \wedge \dom \Psi
				+
				\dom\Psi \wedge \gamma^{(3)} \Psi
			\rp
			-m \bP\Psi \aquatre
	\rp \wedge \omsix \\
	= 
	\eta^{ac}\Omega^i \rho_{i}{}^b_a \wedge \varpi^{(2)}_{bc}
	- 	\frac12 \lp
			\bar\Psi \gamma^{(9)} \wedge \dom \Psi
			+
			\dom\Psi \wedge \gamma^{(9)} \Psi
			\rp
	- m \bP\Psi \vdix
\end{multline*}
with
\[
	\gamma^{(9)} = \gamma^a \vneuf{a}
\]

We want to define a Lagrangian theory on a \emph{structure-less} $10$-manifold $\tot$ such that its solutions provide $\tot$ with the structure of generalised frame bundle and satisfy equations similar to that of the ECD theory~\eqref{eqnos:ECD}. The fields for this theory will be:
\begin{itemize}
	\item A $\so_{1,3}\ltimes\setR^{1,3}$-valued coframe $\varpi = \omega \oplus \alpha$,
	\item A $\Sp$-valued field $\Psi$,
	\item Two extra fields we now introduce.
\end{itemize}

Aside from the (generalised) field equations we want to impose Equation~\eqref{eqno:GFBvarpi}, defining a generalised frame bundle, and a similar horizontality condition on $\Psi$ making it equivariant according to Lemma~\ref{lmno:equivdom}. To this end, we will introduce Lagrange multiplier fields:
\begin{itemize}
	\item A field 
	\[
		P^{BC}_A : \tot \to (\so_{1,3}\ltimes\setR^{1,3})^*\otimes \ExT^2 \lp \so_{1,3}\ltimes\setR^{1,3} \rp
	\]
	which will be required to satisfy the following conditions:
	\begin{enumerate}
		\item $P^{bc}_A = 0$ (which will be replaced below by Equation~\eqref{eqno:Pkappa}),
		\item $P^{BC}_A$ has compact support.
	\end{enumerate}
	
	It will be involved in a term
	\[
		\frac12 P^{BC}_A \lp \d \varpi^A + \fwb\varpi\varpi^A \rp \wedge \vhuit{BC}
	\]
	in the Lagrangian which will impose Equation~\eqref{eqno:GFBvarpi} on $\varpi$ thus defining a generalised frame bundle,
	
	\item A field
	\[
		K^{\alpha i} : \tot \to \bSp\otimes \so_{1,3}
	\]
	of compact support, which will impose the equivariance equation on $\Psi$ using a term
	\[
		\bar K^{\alpha i} \dom \Psi \wedge \vneuf{i}
	\]
\end{itemize}

Note that the term
	\[
		\frac12 P^{BC}_A \lp \d \varpi^A + \fwb\varpi\varpi^A \rp \wedge \vhuit{BC}
	\]
and the term
\[
		\eta^{ac}\Omega^i \rho_{i}{}^b_a \wedge \varpi^{(2)}_{bc}
\]
have the same form. 
As a consequence, we define 
\begin{equation}\label{eqno:defkappa}
	\kappa_A^{bc} = 2\delta_A^i \eta^{bd}\rho_{i}{}^{c}_d
\end{equation}
and replace the condition $P_A^{bc} = 0$ with the condition
\begin{equation}\label{eqno:Pkappa}
	P_A^{bc} = \kappa_A^{bc}
\end{equation}
in order to write
\[
		\frac12 P^{BC}_A \lp \d \varpi^A + \fwb\varpi\varpi^A \rp \wedge \vhuit{BC}	
\]
for the sum of the term with variable $P$'s and the term with $\kappa^{bc}_A$, derived from the usual Einstein-Cartan-Dirac Lagrangian.
We will be using the following notation: $\nonbc{BC}$ to denote indices corresponding to components in $\ExT^2 \so_{1,3} \oplus \so_{1,3}\wedge \setR^{1,3}$, namely components in the supplementary subrepresentation to $\ExT^2 \setR^{1,3}$. For example, $P_A^{\nonbc{BC}}$ are the unconstrained components of $P_A^{BC}$. This Lagrangian (without $\Psi$ terms or $K$ terms) was first introduced in~\cite{LFB}.

The total Lagrangian we will be interested in is the following:
\begin{multline}\label{eqno:Lagtot}
	\Lag[\varpi, \Psi, P, \kappa]
		= \frac12 P^{BC}_A \lp \d \varpi^A + \fwb\varpi\varpi^A \rp \wedge \vhuit{BC} \\
		- m \bP \Psi \vdix
			- \frac12 \lp
				\bar\Psi \gamma^{(9)} \wedge \dom \Psi
				+
				\dom\Psi \wedge \gamma^{(9)} \Psi
				\rp
			+ \frac12 \lp 
				\bar K^{i} \dom \Psi
				-
				\dom\bar\Psi K^i
				\rp	\wedge \vneuf{i}
\end{multline}
(with the contracted spinor $\alpha$ index left implicit).

\subsection{Deriving the Euler-Lagrange equations}

We now derive the Euler-Lagrange equations associated to the Lagrangian~\eqref{eqno:Lagtot}.

\subsubsection*{Useful identities}

Assume in this section $\tot$ is equipped with a $\so_{1,3}\ltimes \setR^{1,3}$-valued $1$-form $\varpi = \omega\oplus \alpha$ which defines a structure of generalised frame bundle.
We gather in this section a couple of identities which will be useful in order to handle the Euler-Lagrange equations.

For this section, let $M$ be any representation of the Lie algebra $\so_{1,3}$.
Let $\Psi^\alpha$ a $M$-valued $k$-form on $\tot$ and $\bP_\alpha$ a $M^*$-valued map on $\tot$.
Recall that we defined $\dom \Psi = \d \Psi + \varpi\cdot \Psi$ with an implicit wedge product. 
Let $\mu$ be any differential form on $\tot$.
We will establish the following formulae:
\begin{subequations}\begin{align}
	\d\vhuit {AB} &= 	\Omega^C\wedge \d\vsept{ABC}-c^C_{AB}\vneuf C
	\\
	\dom \vhuit{bc} &= \Omega^D\wedge \vsept{bcD} \\
	\dom \atrois a &= \Omega^b \wedge \adeux{ab}\\
	\dom \lp \Psi^\alpha\wedge \mu \rp 
		&= \lp \dom \Psi^\alpha \rp \wedge \mu + (-1)^{k} \Psi^\alpha \wedge\d \mu \label{eqno:dlpsimu}
\end{align}\end{subequations}
The first computation uses the \emph{unimodularity} of the Poincaré algebra : $c^B_{AB} = 0$ with $c^{C}_{AB}$ the structure coefficients:
\begin{equation*}
\begin{aligned}
	\d\vhuit {AB}
		&= (\d\varpi^C) \wedge \vsept{ABC}\\
		&= \lp \Omega-\fwb\varpi\varpi \rp^C \wedge \vsept{ABC}\\
		&= \Omega^C \wedge \vsept{ABC} -\frac12 c^C_{DE}\varpi^D \wedge \varpi^E \wedge \vsept{ABC}\\
		&= \Omega^C \wedge \vsept{ABC} - \lp c^C_{BC}\vneuf A - c^C_{AC}\vneuf B + c^C_{AB}\vneuf C \rp \\
		&= \Omega^C \wedge \vsept{ABC} - c^C_{AB}\vneuf C 
\end{aligned}\end{equation*}

Next,
\[
	\omega\cdot \vhuit{bc}
	= -\omega^k c_{bk}^d \vhuit {dc} 
	- \omega^k c_{ck}^d \vhuit{bd}
	= 0
\]
thus 
\[
	\dom \vhuit{bc}
		= \d \vhuit{bc}
		= \Omega^D \wedge \vsept{bcD}
\]

The calculation for $\atrois a$ is very similar:
\begin{equation*}
\begin{aligned}
	\dom \atrois a
		&= \d \atrois a - c^{b}_{ia} \omega^i \wedge \atrois b\\
		&= \d \alpha^b \wedge \adeux{ab} - c^{b}_{ia} \omega^i \wedge \atrois b\\
		&= (\Omega^b - [\omega, \alpha]^b) \wedge \adeux{ab}
			- c^{b}_{ia} \omega^i \wedge \atrois b\\
		&= \Omega^b \wedge \adeux{ab} - c^b_{i d} \omega^i\wedge \alpha^d \wedge \adeux{ab} - c^b_{ia} \omega^i\wedge \atrois b\\
		&= \Omega^b \wedge \adeux{ab} - c^b_{i b} \omega^i \wedge \atrois{a}\\
		&= \Omega^b \wedge \adeux{ab}
\end{aligned}
\end{equation*}

For $\mu$ a differential form, we obtain the following
\begin{equation*}
	\dom \lp \Psi^\alpha\wedge \mu \rp = 
		  \d \lp \Psi^\alpha \wedge \mu \rp + \varpi \cdot \Psi^\alpha \wedge\mu 
		= \lp \dom \Psi^\alpha \rp \wedge \mu + (-1)^{k} \Psi^\alpha \wedge\d \mu
\end{equation*}

Next we check the compatibility with invariant contractions:
\begin{equation*}
\begin{aligned}
	\d(\bP_\alpha \Psi^\alpha)
		&= (\d \bP_\alpha) \Psi^\alpha + \bP_\alpha \d \Psi^\alpha\\
		&= (\dom \bP-\varpi \cdot \bP)_\alpha \Psi^\alpha + \bP_\alpha (\dom \Psi-\varpi \cdot \Psi)^\alpha\\
		&= (\dom \bP_\alpha) \Psi^\alpha + \bP_\alpha (\dom \Psi)^\alpha
			- (\varpi \cdot \bP)_\alpha \Psi^\alpha) + \bP_\alpha (\dom \Psi-\varpi \cdot \Psi)^\alpha \\
		&= (\dom \bP_\alpha) \Psi^\alpha + \bP_\alpha (\dom \Psi)^\alpha
\end{aligned}\end{equation*}
In particular, for $\Psi$ a $M$-valued $k$-form over $\tot$ and $\bp$ a $M^*$-valued $l$-form, the following holds
\begin{equation}
	\d(\bp_\alpha \wedge \Psi^\alpha) = 
	(\dom \bp_\alpha) \wedge \Psi^\alpha + (-1)^l \bp_\alpha \wedge (\dom\Psi)^\alpha
\end{equation}


%
We are now ready to derive the Euler-Lagrange equations.

\subsubsection*{Variations of $P_A^{\nonbc{BC}}$}

Let us derive the Euler-Lagrange equations corresponding to variations of the unconstrained $P_A^{\nonbc{BC}}$. Since the Lagrangian has an order $0$, affine dependency in $P_A^{\nonbc{BC}}$, the equations are straightforwardly derived:
\[
	\lp \d \varpi^A + \fwb\varpi\varpi^A \rp \wedge \vhuit{\nonbc{BC}} = 0
\]
or, more explicitly
\begin{align*}
	\lp \d \varpi^A + \fwb\varpi\varpi^A \rp \wedge \vhuit{bj} &= 0\\
	\lp \d \varpi^A + \fwb\varpi\varpi^A \rp \wedge \vhuit{ij} &= 0	
\end{align*}
Since $(\varpi^a, \varpi^i)$ form a basis of the $1$-forms, $(\vhuit{BC})$ form a basis of the $8$-forms on $\tot$ and as a consequence the equations are equivalent to the existence of coefficients $\Omega^A_{bc}$ such that
\begin{equation}
	\d \varpi^A + \fwb\varpi\varpi^A = \frac12\Omega^A_{bc}\alpha^b\wedge\alpha^c
\end{equation}
Namely, the coframe $\varpi^A$ \emph{endows $\tot$ with the structure of a generalised frame bundle}. 
For the derivation of the subsequent Euler-Lagrange equations, we will assume that this holds.
\footnote{This is in fact not necessary for the remainder of the present section but will allow talking about equivariance.}

\subsubsection*{Variations of $K^{\alpha i}$}

Similarly to the case of $P_A^{\nonbc{BC}}$, the Euler-Lagrange equations with respect to $K^{\alpha i}$ are straightforward to obtain:
\begin{equation*}
	\dom \Psi \wedge \vneuf i = 0
\end{equation*}
namely, since $(\vneuf i, \vneuf a)$ form a basis of the $9$-forms on $\tot$, there exists coefficients $S^\alpha_a \in \Sp\times{\setR^{1,3}}^*$ such that
\begin{equation}
	\dom \Psi^\beta = S^\beta_a \alpha^a
\end{equation}
According to Lemma~\ref{lmno:equivdom}, this is equivalent to $\Psi$ being $\so_{1,3}$-equivariant.

\subsubsection*{Variations of $\Psi$}

\com{Utiliser une notation $\bar\Phi$ ou plutôt une notation $\delta\bar\Psi$ ?}

Since $\Psi$ is an arbitrary map with value in the vector space $\Sp$, a variation can be represented by an arbitrary map $\Phi : \tot \to \Sp$.
The ($\setC$-antilinear) variation of $\Lag[\varpi, \Psi, P, K]$ can be computed as follows: 
\[\begin{aligned}
	2D_{\bP} \Lag[\varpi, \Psi, P, K] \cdot \bar\Phi
		&= 
				-\bar\Phi \gamma^{(9)} \wedge \dom \Psi
				-
				\dom\bar\Phi \wedge \gamma^{(9)} \Psi
			- 	\dom \bar\Phi K^{i}
				\wedge \vneuf{i}
			- 2m\bar\Phi\Psi \vdix \\
		&= 				 -\bar\Phi \gamma^{(9)} \wedge \dom \Psi
				 - \d \lp \bar\Phi \gamma^{(9)}  \Psi \rp
				 + \bar \Phi \dom\lp \gamma^{(9)} \Psi \rp\\
			&\qquad\qquad
			-	\d \lp \bar \Phi K^{i} 
				 \vneuf{i}\rp
				 	+ \bar\Phi \dom (K^{i} \vneuf{i})
			- 2m\bar\Phi\Psi \vdix 
\end{aligned}\]
As a consequence, the Euler-Lagrange equation associated to the variation $\bar\Phi$ of $\Psi$ is
\begin{equation}
	-\frac12 \gamma^{(9)} \wedge \dom \Psi 
		+ \frac12 \dom \lp \gamma^{(9)} \Psi\rp
		- m \Psi
		+ \frac12 \dom \lp \bar K^{i} \vneuf i \rp
		= 0
\end{equation}

\subsubsection*{Variations of $\varpi^A$}

As it is a $\so_{1,3}\ltimes \setR^{1,3}$-valued coframe, $\varpi^A$ is a $\so_{1,3}\ltimes \setR^{1,3}$-valued $1$-form satisfying an open non-degeneracy condition. Therefore, its variations can be represented as a $1$-form $\epsilon^A \in \Omega^1\lp \tot, \so_{1,3}\ltimes \setR^{1,3} \rp$. 
We want the Euler-Lagrange equation corresponding to such a variation of $\varpi^A$.

First, we need to compute the variation of the forms $\vneuf A$ and $\vhuit{AB}$. For this, we compute the variation of $\vdix$:
\[
	D_\varpi \vdix \cdot \epsilon = \epsilon^A \wedge \vneuf A
\]
\com{Notation $D_\varpi$ ??}

In order to compute $D_\varpi \vneuf A$, we use the following identity:
\[
	\varpi^B \wedge \vneuf A = \delta^B_A \vdix 
\]
which varies as follows under first order variations of $\varpi^B$:
\[\begin{gathered}
	\epsilon^B \wedge \vneuf A + \varpi^B \wedge D_\varpi \vneuf A \cdot \epsilon 
		= \delta^B_A \epsilon^C \wedge \vneuf C\\
	\varpi^B \wedge D_\varpi \vneuf A \cdot \epsilon
		= \delta^B_A \epsilon^C \wedge \vneuf C
			- \epsilon^B \wedge \vneuf A
		= \varpi^B \wedge \epsilon^C \wedge \vhuit{AC}
\end{gathered}
\]
Since $\varpi$ is non-degenerate, this implies that
\begin{equation}
	D_\varpi \vneuf A \cdot \epsilon
		= \epsilon^C \wedge \vhuit{AC} 
\end{equation}
In particular, we obtain the variation of $\gamma^{(9)}$:
\[
	\gamma^a\epsilon^C \wedge \vhuit{aC}
\]

Similarly, considering the variation of $\varpi^C\wedge \varpi^D \wedge \vhuit {AB}$ one obtains the expression
\begin{equation}
	D_\varpi \vhuit {AB} \cdot \epsilon
		= \epsilon^C \wedge \vsept{ABC} 	
\end{equation}

We can now compute the variation of $\Lag[\varpi, \Psi, P, K]$ under the variation $\epsilon$ of $\varpi$. We make two separate computations for the $P$ term and the $\Psi$ term. Let us introduce the following notation for the action of $\so_{1,3}$ on $\Sp$:
\[
	\sigma_i : \so_{1,3} \to \End(\Sp)
\]
They are related to the $\sigma^a_b:= \frac14 [\gamma^a, \gamma^{b'}]\eta_{b'b}$ mentioned in Section~\ref{secno:EC} as follows:
\[
	\frac12 {\rho_i}^b_a \sigma_b^a = \sigma_i
\]
In particular, the $\sigma_i$ act on $\Sp$ by anti-hermitian endomorphisms.

The variation of the $P$ term is
\begin{equation}\begin{multlined}
	\frac12 P^{BC}_A\lp 
		\lp
			\d \epsilon^A + \wb\epsilon\varpi^A 		
		\rp
		\wedge \vhuit{BC}
		+ \Omega^A \wedge \epsilon^D \wedge \vsept{BCD} \rp\\
	= \frac12 P^{BC}_A \lp \dom \epsilon^A \wedge \vhuit{BC} + \epsilon^D \wedge \Omega^A \wedge \vsept{BCD} \rp\\
	= \d \lp \frac12 P^{BC}_A \epsilon^A \wedge \vhuit{BC} \rp
		+ \epsilon^A \wedge \dom \lp \frac12 P^{BC}_A \vhuit{BC} \rp
		+ \frac12 P^{BC}_A \epsilon^D \wedge \Omega^A \wedge \vsept{BCD}
\end{multlined}
\end{equation}

The variation of the term 
$ 			- \frac12 \lp
				\bar\Psi \gamma^{(9)} \wedge \dom \Psi
			- 
				\bar K^{i} \dom \Psi \wedge \vneuf{i} 
				\rp
$
is
\begin{equation}
\begin{multlined}
	-\frac12 \bP \gamma^a \epsilon^B \wedge \vhuit{aB} \wedge \dom \Psi
	-\frac12 \bP \gneuf \epsilon^i \sigma_i \Psi
	+\frac12 \bar K^i \epsilon^j\sigma_j \Psi \wedge \vneuf i
	+ \frac12\bar K^i \dom \Psi \wedge \epsilon^B \wedge \vhuit{iB}\\
	= -\frac12 \epsilon^B \wedge \lp
		\bP \gamma^a \dom \Psi \wedge \vhuit{aB}
		- \delta^j_B \bP \sigma_j \gneuf \Psi
		+ \bar K^i \lp - \delta^j_D\sigma_j \Psi \vneuf i
		+ \dom \Psi \wedge \vhuit{iB} \rp
	\rp	
\end{multlined}
\end{equation}
The variation of the complex conjugate term is then the complex conjugate variation.

The variation of $m\bP \Psi \vdix$ is
\[
	m\bP \Psi \epsilon^A \wedge \vneuf A
\]

The Euler-Lagrange equation associated to $\Lag[\varpi, \Psi, P, K]$ with respect to variations of $\varpi^D$ is therefore:
\begin{equation}
\begin{multlined}
	\dom\lp \frac12 P^{BC}_D \vhuit{BC} \rp 
	+ \frac12 P^{BC}_A \Omega^A \wedge \vsept{BCD}
	+ \frac12 \delta^i_D \bP \{\sigma_i, \gneuf \} \Psi 
	- m \bP\Psi \vneuf D \\
	- \frac12 \lp \bP \gamma^a \dom \Psi \wedge \vhuit{aD} 
			- \delta^i_D \bP \gneuf \sigma_i \Psi
			+ \bar K^i \lp - \delta^j_D\sigma_j \Psi \vneuf i
				+ \dom \Psi \wedge \vhuit{iB} \rp
			\rp\\
	+ \frac12 \lp \dom \bP \gamma^a \Psi \wedge \vhuit{aD} 
			+ \delta^i_D \bP \sigma_i \gneuf \Psi
			- \lp \delta^j_D\bP \vneuf i - \dom\bP \vhuit{iD} \rp K^i
			\rp = 0
\end{multlined}
\end{equation}
using the anticommutator notation $\{X, Y\}:= XY + YX$.

\begin{proposition}[Euler-Lagrange equations]
	The Euler-Lagrange equations associated to the Lagrangian~\eqref{eqno:Lagtot} are the following:
	\begin{subequations}\label{eqnos:ELeqs}
	\begin{equation}\label{eqno:ELGFB}
		\d \varpi^A + \fwb\varpi\varpi^A = \frac12\Omega^A_{bc}\alpha^b\wedge\alpha^c
	\end{equation}
	
	\begin{equation}\label{eqno:ELspinequiv}
		\dom \Psi^\beta = S^\beta_a \alpha^a
	\end{equation}

	\begin{equation}\label{eqno:ELPsi}
	-\frac12 \gamma^{(9)} \wedge \dom \Psi 
		+ \frac12 \dom \lp \gamma^{(9)} \Psi\rp
		-m \Psi \vdix
		+ \frac12 \dom \lp \bar K^{i} \vneuf i \rp
		= 0
	\end{equation}

	\begin{equation}\label{eqno:ELvarpi}
	\begin{multlined}
		\dom\lp \frac12 P^{BC}_D \vhuit{BC} \rp 
		+ \frac12 P^{BC}_A \Omega^A \wedge \vsept{BCD}
		+ \frac12 \delta^i_D \bP \{\sigma_i, \gneuf \} \Psi
		- m\bP \Psi \vneuf D \\
	- \frac12 \lp \bP \gamma^a \dom \Psi \wedge \vhuit{aD} 
			+ \bar K^i \lp - \delta^j_D\sigma_j \Psi \vneuf i
				+ \dom \Psi \wedge \vhuit{iD} \rp
			\rp\\
	+ \frac12 \lp \dom \bP \gamma^a \Psi \wedge \vhuit{aD} 
		- \lp \delta^j_D\bP \sigma_j \vneuf i
			- \dom\bP \wedge \vhuit{iD} \rp K^i
		    \rp = 0
	\end{multlined}		
	\end{equation}
	\end{subequations}
\end{proposition}

In the next section we present the mechanism which allows to interpret these equations as a generalisation of the ECD field equations. In particular, there is a topological cancellation mechanism which is crucial in order to decouple the physical fields from the Lagrange multipliers $P^{BC}_A$ and $K^i$.

%







\section{Derivation of the classical field equations}\label{secno:EuclFE}
In this section, we show how the ECD field equations~\eqref{eqnos:ECD} can be derived from the Euler-Lagrange equations~\eqref{eqnos:ELeqs}. This will be done under the assumption that the generalised frame bundle structure defined by $\varpi^A$ is a classical frame bundle. In this sense, the Euler-Lagrange equations~(\ref{eqno:ELPsi},\ref{eqno:ELvarpi}) give a generalisation of the classical field equations to the framework of generalised frame bundles.

\subsection{Preliminary manipulations}

Our analysis will allow decoupling the \enquote{physical fields} $\varpi, \Psi$ from the Lagrange multipliers $P, K$. This is necessary since the Lagrange multipliers are coupled to \emph{differential terms} as opposed to the simpler case of holonomic constraints (of type $f=0$).

To this aim, we will not use the Euler-Lagrange form of the variations of $\Lag$. Let us separate the $P^{\nonbc{BC}}_A$ terms from the $\kappa^{bc}_i$ terms in $P^{BC}_A \dom \epsilon^A \wedge \vhuit{BC}$:
\[\begin{aligned}
	\frac12 P^{BC}_A \dom \epsilon^A \wedge \vhuit{BC} 
		&= \frac12 P^{\nonbc{BC}}_A \dom \epsilon^A \wedge \vhuit{\nonbc{BC}} + \frac12 \kappa^{bc}_i \dom \epsilon^i \wedge \vhuit{bc}\\
		&= \frac12 P^{\nonbc{BC}}_A \dom \epsilon^A \wedge \vhuit{\nonbc{BC}} + \d \lp \frac12 \kappa^{bc}_i \dom \epsilon^i \wedge \vhuit{bc} \rp +  \epsilon^i \wedge \dom \lp \frac12\kappa^{bc}_i \vhuit{bc} \rp\\
		&= \frac12 P^{\nonbc{BC}}_A \dom \epsilon^A \wedge \vhuit{\nonbc{BC}} + \d \lp \frac12 \kappa^{bc}_i \dom \epsilon^i \wedge \vhuit{bc} \rp +  \epsilon^i \wedge \frac12\kappa^{bc}_i \dom \vhuit{bc} \\
		&= \frac12 P^{\nonbc{BC}}_A \dom \epsilon^A \wedge \vhuit{\nonbc{BC}} + \d \lp \frac12 \kappa^{bc}_i \dom \epsilon^i \wedge \vhuit{bc} \rp +  \epsilon^i \wedge \frac12\kappa^{bc}_i \Omega^D \wedge \vsept{bcD}
\end{aligned}\]
We used $\dom \kappa^{bc}_i = 0$ which is a consequence of $\kappa^{bc}_i$ being constants which are $\so_{1,3}$-invariant (defined at~\eqref{eqno:defkappa}).

Instead of the Euler-Lagrange equations~(\ref{eqno:ELPsi},\ref{eqno:ELvarpi}), we will use the following equations with differential terms:
\begin{subequations}\label{eqnos:EL2}
\begin{gather}
	\begin{multlined}
		\frac12P^{\nonbc{BC}}_A \dom \epsilon^A \wedge \vhuit{BC} 
		+ \epsilon^i \wedge \frac12\kappa^{bc}_i \Omega^D\wedge \vsept{bcD}
		+ \frac12P^{BC}_A \epsilon^D \wedge \Omega^A \wedge \vsept{BCD}\\
		-\frac12 \epsilon^D \wedge \lp 
			\lp \bP \gamma^a \dom \Psi - \dom \bP \gamma^a \Psi \rp\wedge \vhuit{aD} 
			+ \lp \bar K^i \dom\Psi	- \dom \bP K^i \rp\wedge \vhuit{iD}
				- \delta^j_D \lp \bar K^i \sigma_j \Psi - \bP \sigma_j \bar K^i \rp\vneuf i
			\rp		
			\\
		+ \frac12 \epsilon^i\wedge \bP \{\sigma_i, \gneuf \} \Psi 
		- \d \lp \frac12P^{\nonbc{BC}}_A \epsilon^A \wedge \vhuit{\nonbc{BC}} \rp
		= 0
	\end{multlined}	
\end{gather}
\begin{gather}
	-\bar \Phi \lp \frac12 \gneuf \wedge \dom \Psi - \frac12\dom \lp \gneuf \Psi \rp 
	+ m\Psi \vdix \rp
	- \dom \bar \Phi K^i \wedge \vneuf i 
	= - \d \lp \bar\Phi K^i\wedge \vneuf i \rp
\end{gather}
\end{subequations}

\subsection{Restricting the variations}

We want to obtain fields equations from Equations~(\ref{eqnos:ELeqs},\ref{eqnos:EL2}).
We first assume that Equations~(\ref{eqno:ELGFB}, \ref{eqno:ELspinequiv}) are satisfied, so that $\varpi^A$ defines a generalised frame bundle structure on $\tot$.
In particular, both $\Omega^A$ and $\dom \Psi$ are horizontal.

\textbf{We now assume that this generalised frame bundle structure is that of a classical \emph{spinor} frame bundle $\tot \to \ET$.} We will show that in this case, we can obtain the ECD field equations on $\ET$.

\begin{remark}
	From a topological perspective, this is equivalent to requiring that the action of $\so_{1,3}$ on $\tot$ is complete, that it integrates to a \emph{free} action of $\Spinp_{1,3}$ and that this action is proper (\cite{DGMP1}, Corollary 6.5.1).
\end{remark}

We will restrict the considered variation $\epsilon, \Phi$ in two ways:
\begin{enumerate}
	\item $\epsilon^A$ will be required to be \emph{horizontal}: $\epsilon^A = \epsilon^A_b \alpha^b$. This corresponds to \emph{restricting} to a subset of the Euler-Lagrange equations.
	
	\item $\epsilon^A$ and $\Phi$ will be (provisionally) required to be \emph{$\Spinp_{1,3}$-equivariant}. It will turn out that the obtained equations manifestly extend 
	\footnote{Since the variational derivative of $\Lag$ in the direction $(\epsilon, \Phi)$ is $\Cinf(\tot)$-linear in $(\epsilon, \Phi)$ and given that $\Spinp_{1,3}$-equivariant maps generate all maps over $\Cinf(\tot)$, restricting to equivariant variations does not weaken the set of equations.}
	to non-equivariant $\epsilon^A, \Phi$.
\end{enumerate}

As a consequence of these hypotheses, both $\dom \epsilon^A$ and $\dom \Phi$ are purely horizontal (Lemma~\ref{lmno:equivdom}). Therefore we obtain the following identities:
\begin{gather*}
	\frac12 P^{BC}_A \dom \epsilon^A \wedge \vhuit{BC}
		= \frac12 \kappa^{bc}_A \dom \epsilon^A \wedge \vhuit{bc}\\
	\frac12P^{BC}_A \epsilon^D \wedge \Omega^A \wedge \vsept{BCD}
	=
	\frac12P^{bc}_A \epsilon^d \wedge \Omega^A \wedge \vsept{bcd}\\
	\bar K^i \dom \Psi \wedge \vhuit{iB} = 0\\
	\dom \bP K^i\wedge \vhuit{iD} = 0\\
	\bar K^i \dom\Phi \wedge \vneuf i = 0\\
%
%
%
	\epsilon^D \wedge \dom \Psi \wedge \vhuit{iD} = 0\\
	\epsilon^D \wedge \dom \bP \wedge \vhuit{iD} = 0\\
\end{gather*}

We thus obtain the following equations:

\begin{subequations}\label{eqnos:EL3}
\begin{gather}
	\begin{multlined}
		\frac12 \kappa^{bc}_i \lp \epsilon^i \wedge \Omega^D \wedge \vsept{bcD} + \epsilon^d \wedge \Omega^i \wedge \vsept{bcd} \rp 	
		-\frac12 \epsilon^b \wedge \lp 
			\bP \gamma^a \dom \Psi
			- \dom \bP \gamma^a \Psi  \rp \wedge \vhuit{ab} \\
		+ \frac12 \epsilon^i \wedge \bP \{\sigma_i, \gneuf \} \Psi
		- m\bP \Psi \epsilon^b \wedge \vneuf b
		= \d \lp \frac12P^{\nonbc{BC}}_A \epsilon^A \wedge \vhuit{\nonbc{BC}} \rp
	\end{multlined}	
\end{gather}
\begin{gather}
	\frac12 \bar\Phi \lp -\gneuf \wedge \dom \Psi 
		+ \dom \lp \gneuf \Psi \rp \rp
		- \frac12(\dom \bar\Phi)  K^{i} \vneuf i
		- m \bar\Phi \Psi \vdix
		= -\frac12 \d \lp \bar \Phi  K^{i} \vneuf i \rp
\end{gather}
\end{subequations}

In order to decouple the two sides, we will need on the \emph{cancellation mechanism} described in the next section.

\subsection{The cancellation mechanism}

The cancellation mechanism uses integration on $\Spinp_{1,3}$-orbits in order to decouple the left and right sides of Equations~\eqref{eqnos:EL3}. It relies on the following Lemma:
\begin{lemma}\label{lmno:cancellemma}
	Let $\tot$ be a manifold with an action of $\Spinp_{1,3}$. Denote by $A$ the ideal of differential forms vanishing on each orbit of $\Spinp_{1,3}$. Let $E\in \Omega^6(\tot)$ and $P\in \Omega^5(\tot)$ such that:
	\begin{itemize}
	\item $E$ is $\Spinp_{1,3}$-invariant,
	\item $P$ has compact support on every $\Spinp_{1,3}$-orbit,
	\item $ E \equiv \d P \bmod A$.
	\end{itemize}
	
	Then $E \equiv \d P \equiv 0 \bmod A$.
\end{lemma}
\begin{proof}
	Let $O$ be a $\Spinp_{1,3}$-orbit of $\tot$; we see it as an immersed submanifold. Since $\d P$ has compact support on $O$, so does $E$. They are thus integrable and 
	\[
		\int_O E = \int_O \d P = 0
	\]
	However, since $E$ is $\Spinp_{1,3}$-invariant, it is either identically $0$ or defines a volume form on $O$. Thus $E\in A$.
\end{proof}

\begin{remark}
	It is possible to replace the group action of $\Spinp_{1,3}$ by a Lie algebra action of $\spin_{1,3}$ since the orbits always form immersed submanifolds~(\cite{MichorTopics}, Section I.3).
\end{remark}

We want to turn Equations~\eqref{eqnos:EL3} into an equation of the form involved in Lemma~\ref{lmno:cancellemma}. In particular, we want $6$-forms. We will be using the following notations
\begin{align}
	\Omega^A &= \frac12\Omega^A_{bc} \alpha^b \wedge \alpha^c\\
	\dom \Psi^\beta &= S^{\beta}_a \alpha^a\\
	\epsilon^A &= \epsilon^A_b \alpha^b\\
	\dom \Phi^\beta &= F^\beta_a \alpha^a
\end{align} 
as well as a bracket notation $[abc]$ to denote \emph{normalised} antisymmetrizing of the indices. Furthermore, note that $\Psi$ is $\Spinp_{1,3}$ equivariant and $\gneuf$ is as well, as a contraction between the $\Spinp_{1,3}$-invariant $\gamma^a$ and the equivariant $\vneuf a$. As a consequence, $\gneuf \Psi$ is equivariant and
\begin{equation}
	\dom \lp \gneuf \Psi \rp = M^\beta \vdix
\end{equation}
with $M^\beta$ which are $\Spinp_{1,3}$-equivariant.

The idea is to factor out of a factor $\aquatre$ from Equations~\eqref{eqnos:EL3} in order to obtain equations on $6$-forms which hold on orbits. We will handle the different terms separately:
\begin{gather*}
	\epsilon^i \wedge \Omega^D \wedge\vhuit{bcD}
		= \epsilon^i_{[b} \Omega^d_{cd]} \omsix\wedge \aquatre\\
	\epsilon^d \wedge \Omega^i \wedge \vsept{bcd}
		= \epsilon^d_{[b} \Omega^i_{cd]} \omsix\wedge \aquatre\\
	\epsilon^b \wedge \dom \Psi \wedge \vhuit{ab}
		= \epsilon^b_{[a} S_{b]} \omsix \wedge \aquatre\\
	\d \lp P^{\nonbc{BC}}_A \epsilon^A \wedge \vhuit{\nonbc{BC}} \rp 
		= \d \lp 4 P^{jc}_A \epsilon^A_{c} \vneuf{j} \rp
		= 4 \d \lp P^{jc}_A \epsilon^A_c \omcinq{j} \rp \wedge \aquatre\\
	\gneuf \wedge \dom \Phi = \gamma^a F_a \omsix\wedge \aquatre\\
	\gneuf \wedge \dom \Psi^\beta = \gamma^a S_a^\beta \omsix\wedge \aquatre\\
	\d \lp \gneuf \Psi \rp^\beta = M^\beta \omsix\wedge \aquatre\\
	\epsilon^i \wedge \{\sigma_i, \gneuf\} = \epsilon^i_a \{\sigma_i, \gamma^a\} \omsix \wedge \aquatre\\
	\epsilon^b \wedge \vneuf b = \epsilon^b_b \omsix \wedge \aquatre
\end{gather*}
We used $\d \aquatre = 0$ which is a consequence of~\eqref{eqno:ELGFB}. Note that the $M^\beta$ are non independent from the $S_a^\beta$.

We thus obtain the following equations
\begin{subequations}\label{eqnos:ELmoda}
\begin{gather}
	\begin{multlined}
		\frac12 \kappa^{bc}_i \lp \epsilon^i_{[b} \Omega^d_{cd]} + \epsilon^d_{[b} \Omega^i_{cd]} \rp \omsix	
		-\frac12 \epsilon^b_{[a} \lp 
			\bP \gamma^a  S_{b]} 
			- \bar S_{b]}  \gamma^a \Psi \rp \omsix \\
		+ \frac12 \epsilon^i_a \bP \{\sigma_i, \gamma^a \} \Psi \omsix
		- \epsilon^b_b m\bP \Psi \omsix
		= 2 \d \lp P^{jc}_A \epsilon^A_c \omcinq j \rp 
		\bmod \alpha^a
	\end{multlined}	
\end{gather}
\begin{gather}
	 -\frac12 \bar \Phi S_a \gamma^a \omsix  +\frac12 \bar \Phi M \omsix -m \bar\Phi \Psi \omsix
	 	= -\d \lp \bar\Phi K^i \omcinq i\rp 
			\bmod \alpha^a
\end{gather}
\end{subequations}

These equations hold modulo $\alpha^a$, which vanishes on each orbit of $\Spinp_{1,3}$. Applying Lemma~\ref{lmno:cancellemma}, we conclude that 
\begin{subequations}\label{eqnos:?}
\begin{gather}
	\begin{multlined}
		\frac12 \kappa^{bc}_i \lp \epsilon^i_{[b} \Omega^d_{cd]} + \epsilon^d_{[b} \Omega^i_{cd]} \rp \omsix	
		-\frac12 \epsilon^b_{[a} \lp 
			\bP \gamma^a  S_{b]} 
			- \bar S_{b]}  \gamma^a \Psi \rp \omsix
		+ \frac12 \epsilon^i_a \bP \{\sigma_i, \gamma^a \} \Psi \omsix 
		- m \bP \Psi \omsix
		\\
		\equiv 2 \d \lp P^{jc}_A \epsilon^A_c \omcinq j \rp 
		\equiv 0
		\bmod \alpha^a
	\end{multlined}	
\end{gather}
\begin{gather}
	 -\frac12\bar \Phi S_a \gamma^a \omsix  -\frac12 \bar \Phi M \omsix -m \bar\Phi \Psi \omsix \equiv 0
	\bmod \alpha^a
\end{gather}
\end{subequations}
As a consequence
\begin{subequations}\label{eqnos:??}
\begin{gather}\label{eqno:ELepsilontens}
	\begin{multlined}
		\frac12 \kappa^{bc}_i \lp \epsilon^i \wedge \Omega^d \wedge \vsept{bcd} + \epsilon^d \wedge \Omega^i \wedge \vsept{bcd} \rp 	\\
		-\frac12 \epsilon^b \wedge \lp 
			\bP \gamma^a \dom \Psi \wedge \vhuit{ab}
			- \dom \bP \gamma^a \Psi \wedge \vhuit{ab} \rp\\
		+ \frac12 \epsilon^i \wedge \bP \{ \sigma_i, \gneuf \} \Psi
		- m \bP \Psi \epsilon^b \wedge\vneuf b
		= 0
		\end{multlined}	
\end{gather}
\begin{gather}
		\d \lp \frac12P^{\nonbc{BC}}_A \epsilon^A \wedge \vhuit{\nonbc{BC}} \rp
		=0
\end{gather}
\begin{gather}\label{eqno:ELPhitens}
		\bar \Phi \lp \frac12 \gneuf \wedge \dom \Psi - \frac12 \dom \lp \gneuf \Psi \rp + m \Psi \vdix\rp = 0
\end{gather}
\end{subequations}
hold for every \emph{equivariant and horizontal} $\epsilon^A$ and every \emph{equivariant} $\Phi$. However since these expressions are tensorial in $\epsilon^A$ and $\Phi$, these equations hold for \emph{any} $\Phi$ and \emph{any} horizontal $\epsilon^A$.

Therefore we obtain the following equations:
\begin{subequations}
\begin{gather}
	\frac12 \kappa^{bc}_i \wedge \Omega^d \wedge \vsept{bcd} + \frac12 \bP \{\sigma_i, \gneuf \} \Psi = 0\\
	\frac12 \kappa^{cd}_i \Omega^i \wedge \vsept{bcd}
	- \frac12 \lp 
		\bP \gamma^a \dom \Psi \wedge \vhuit{ab}
		- \dom \bP \gamma^a \Psi \wedge \vhuit{ab}
		\rp
	-m \bP \Psi \vneuf b= 0\\
	\gneuf \wedge \dom \Psi - \dom\lp \gneuf \Psi \rp + m\Psi \vdix = 0
\end{gather}
\end{subequations}

In order to obtain equations on the base manifold $\ET$, we want to factor out the factor $\omsix$ we added in Section~\ref{secno:Lagrangian} so as to build the Lagrangian $10$-form. It is straightforward in the first two equations: $\vsept{bcd} = \alpha^{(1)}_{bcd}\wedge \omsix$ and $\vhuit{ab} = \alpha^{(2)}_{ab} \wedge \omsix$. For the third one, let us decompose
\[
	\gneuf = \gamma^a \vneuf a = \gamma^a \atrois{a} \wedge \omsix
\]
Then
\[\begin{aligned}
	\dom \gneuf
	&= \dom \lp \gamma^a \atrois a \wedge \omsix \rp\\
	&= \dom \lp \gamma^a \atrois a \rp \wedge \omsix + \gamma^a \atrois a \wedge \d \omsix \\
	&= \dom \lp \gamma^a \atrois a \rp \wedge \omsix + \gamma^a \atrois a \wedge \d \omega^i \wedge \omcinq i\\
	&= \dom \lp \gamma^a \atrois a \rp \wedge \omsix + \gamma^a \underbrace{\atrois a \wedge \Omega^i}_{=0} \wedge \omcinq i\\
	&=  \dom \lp \gamma^a \atrois a \rp \wedge \omsix
\end{aligned}\]
with $\dom \lp \gamma^a \atrois a \rp$ which is horizontal (Lemma~\ref{lmno:equivdom}). In fact it can be explicitly expressed as follow: 
\begin{equation*}
\begin{aligned}
	\dom \lp \gamma^a \atrois a \rp 
%
		&= \dom\gamma^a \wedge \atrois a + \gamma^a \dom \atrois a\\
		&= 0 + \gamma^a \Omega^b \wedge \adeux{ab}
\end{aligned}
\end{equation*}
We bypassed the proof that the $\omega$ terms cancel out since we know that the result is horizontal.

Therefore, we obtain the following basic $6$-forms equations:
\begin{gather*}
	\frac12 \kappa^{bc}_i \Omega^d \wedge \aun{bcd}
	+ \frac12 \bP \{\sigma_i, \gamma^a \atrois a \} \Psi = 0\\
	\frac12 \kappa^{cd}_i \Omega^i \wedge \aun{bcd}
	- \frac12\lp
		\bP \gamma^a \dom \Psi \wedge \adeux{ab}
		- \dom \bP \gamma^a \Psi \wedge \adeux{ab}
		\rp
	- m \bP \Psi \vneuf b= 0\\
	\frac12 \gamma^a \atrois a \wedge \dom \Psi - \frac12 \dom\lp \gamma^a \atrois a \Psi \rp 
	+m \Psi \aquatre = 0
\end{gather*}
with
\[
	\dom \lp \gamma^a \atrois a \Psi \rp = \Omega^b \wedge \adeux{ab} \gamma^a \Psi - \gamma^a \atrois a \wedge \dom \Psi 
\]

It is possible to reformulate the first equation without the $\kappa_i^{bc}$ factor, since
\[
	\sigma_i = \frac12 \rho_i{}^b_a \sigma_b^b = \frac1{16} \kappa_i^{bc} [\gamma_b, \gamma_c]
\]
and since $\frac12\kappa$ defines an embedding $\so_{1,3} \hookrightarrow \lp \setR^{1,3}\rp^* \otimes \setR^{1,3}$ the equation is equivalent to
\[
	\Omega^d \wedge \aun{bcd}
	+ \frac1{16}\bP \{[\gamma_b, \gamma_c], \gamma^a \atrois a \} \Psi = 0\\
\]

Finally, these equations can be turned into equations on the associated fields on $\ET$ as follows, using a notation $(\Curv^I_J, T^I)$ for the curvature and the torsion $2$-forms and $\psi, \bp$ for the associated spinor and dual spinor fields on $\ET$:
\begin{gather}
	T^L \wedge e^{(1)}_{JKL}
	= - \frac1{16} \bp \{[\gamma_J, \gamma_K], \gtrois\} \psi\\
	\Curv^{KL} \wedge e^{(1)}_{JKL}
	= 
	\frac12 \lp \bp \gamma^I \nabla \psi \wedge e^{(2)}_{IJ}
		- \nabla \bp \gamma^I \psi \wedge e^{(2)}_{IJ}
		\rp
	+ m \bp\psi e^{(3)}_J
	\\
	\gtrois \wedge \nabla \psi - \frac12 \gamma^I \psi T^J \wedge e^{(2)}_{IJ} + m \psi e^{(4)}= 0
\end{gather}
using the convention from Appendix~\ref{anndual} for the dual forms $e^{(n-k)}$ and $\gtrois := \gamma^I e^{(3)}_I$.

These match with a auxiliary bundle $\aux$ version of the ECD Equations~\eqref{eqnos:ECD}. Thus our we have derived the ECD equations from the Euler-Lagrange equations of our model, in the case that the generalised frame bundle is a standard frame bundle. Therefore, these Euler-Lagrange equations can hold as a generalised of the ECD field equations to the framework of generalised frame bundles. Note however that since we required the variations of $\epsilon^A$ to be purely horizontal, we only handled \emph{part} of the Euler-Lagrange equations. In fact our approach does not work for the remaining Euler-Lagrange equations: these have to be analysed through other means, which shall be the subject of future research.

\subsection{A Euclidean variant}

In this final section we present a Euclidean variant of the model. More detail can be found in~\cite{PDMPhD}, Section 8.3.2. Using a global compactness assumption, it is possible to obtain strong conclusions. 

We therefore replace $\so_{1,3}\ltimes \setR^{1,3}$ with $\so_4 \ltimes \setR^4$ and adapt all other definitions. The $\kappa_i^{bc}$ are adapted to the Euclidean signature and now represent the embedding $\so_4 \hookrightarrow \ExT^2 \setR^4$. 

Furthermore, we assume that the $10$-manifold $\tot$ is compact. In this case, given a generalised frame bundle structure on $\tot$, the associated action of $\so_4$ is necessarily integrable into an action of $\Spin_4$, which is then proper.
This implies the existence of a dense open subset $\tot_{prin} \subset \tot$ on which all isotropy groups are identical to a subgroup $S\subset \Spin_4$ which is either trivial or the center subgroup $\{\pm1\}$. This is enough to deduce that $\tot_{prin} \twoheadrightarrow \tot_{prin}/\Spin_4$ defines a $\Spin_4/S$-principal bundle. The generalised frame bundle structure restricts thus to a standard frame bundle structure on an open submanifold on $\tot$.

\begin{theorem}
	Let $\tot$ be a compact $10$-manifold.
	Let:
	\begin{itemize}
	\item $\varpi^A$ be a $\so_4\ltimes \setR^4$-valued coframe on $\tot$,
	
	\item $P_A^{BC}$ be a $\ExT^2 \so_4 \otimes \so_4^*$-valued field which satisfies the constraint $P_A^{bc} = \kappa_A^{bc}$,
	
	\item $\Psi$ be a $\Sp$-valued field on $\tot$,
	
	\item $K^i$ be a $\Sp\otimes \so_4$-valued field on $\tot$.
	\end{itemize} 
	
	Then if $(\varpi^A, P^{BC}_A, \Psi, K)$  satisfy the EL equations for the following Lagrangian:
	\[\begin{aligned}
		\Lag[\varpi, P, \Psi, K]
			&= \frac12 P^{BC}_A \vhuit{BC} \wedge \lp \d \varpi + \fwb\varpi\varpi \rp^A
	\end{aligned}
	\]
	Then
	\begin{itemize}
	\item There is an action of $\Spin_4$ on $\tot$ for which $\varpi$ is \textbf{normalised and equivariant},
	\item The action of $\Spin_4$, or possibly of the quotient $\SO_4$, is \textbf{free} on a dense open submanifold $\tot_{prin}$. 
	\\
	If $\Psi$ takes somewhere a nonzero value then the group acting freely in $\Spin_4$,
	\item $\tot_{prin}$ defines a \textbf{Riemannian Spin structure} on $\ET:=\tot_{prin}/\Spin_4$,
	
	\item $\Psi$ has an associated spinor field $\psi$ on $\ET$ which along with the Riemannian structure satisfies the \textbf{Einstein-Cartan-Dirac equations}
	\end{itemize}
\end{theorem}

\paragraph{Acknowledgements}
The present work was carried out by the author during the preparation of their PhD at the Universit\'e Paris Cit\'e, France.

\paragraph{Declaration of Interests} None

\appendix\clearpage

\section{Conventions for dual forms}\label{anndual}
Let $E$ a $n$-dimensional $\setK$-vector space with $n\in\mathbb N$. Provide it with a volume element $\vol\in\ExT^n E$. The volume element defines an isomorphism $\ExT^n E \simeq \setK$. Note that in the article we would use a space (bundle) of \emph{linear forms} as $E$. This Appendix will make use of the Einstein summation convention for repeated indices (as stated in Section~\ref{secno:notconv}).

We define here the notation for interior products between $p$-covectors and $q$-vectors. We will write $\wedge_{i:1 \to q} \alpha^i$ for the wedge product $\alpha^1\wedge\alpha^2\dots\wedge\alpha^q$. Let $\alpha\in\ExT^q E$ and $(X_j)\in (E^*)^p$. We will use the following convention : 
\begin{equation}
\boxed{
	(X_1\wedge X_2\wedge \dots \wedge X_p)\lr \alpha
	= i_{X_1\wedge X_2\wedge \dots \wedge X_p} \alpha
	= i_{X_p}i_{X_{p-1}}\cdots i_{X_1}\alpha}
\end{equation}
The volume element provides isomorphisms $\ExT^p E^* \to \ExT^{n-p} E$ under the following contraction
\begin{equation}
	X\in\ExT^p E^* \mapsto X\lr\vol
\end{equation}
Note that if $E$ is provided with an inner product then precomposing the isomorphisms by the induced inner product on the exterior powers $\ExT^p E \to \ExT^p E^*$ gives the \emph{Hodge duality operator} $\star$.

Let $(e^i)_{\lel 1 i n}$ be a \emph{direct} basis of $E$ and $(u_i)$ its dual basis. Let $I$ a sequence of $p$ indices in $\llbracket 1,n\rrbracket$. Define 
\[ e^I = \operatornamewithlimits{\bigwedge}_{i : 1 \to p} e^{I_i} \]
and dually
\[ u_I = \operatornamewithlimits{\bigwedge}_{i : 1 \to p} u_i \]
We will explicitly use the map $\ExT^p E^* \to \ExT^{n-p} E$ and adopt the notation
\begin{equation}
	e^{(n-p)}_I := u_I\lr\vol
\end{equation}

In components, $\vol$ is represented by the \emph{Levi-Civita symbol} commonly written $\ve_ {i_1\dots i_n}$, interpreted as a \emph{completely antisymmetric} rank $n$ tensor \emph{in a basis of determinant $1$}. For example, the duality between $E^*$ and $\ExT^{n-1} E$ is expressed, for $A=A^iu_i \in E^*$ as follows : 
\begin{align}
	(A^iu_i)\lr \vol &= A^i(u_i\lr\vol) = A^ie^{(n-1)}_i\\ 
	(A\lr \vol)_{i_1\dots i_{n-1}} &= A^{i_0}\ve_{i_0i_1\dots i_{n-1}}
\end{align}
which generalizes to the case of $p$-forms on $E$ in a straightforward manner.

From now on we assume that $I$ does not contain repeated indices of the basis so that $u_I$ and $e^I$ are nonzero. The conventions have been chosen such that for two increasing multi-indices $I$ and $J$,
\begin{equation}
u_I\lr e^J = \delta^J_I
\end{equation}
Define $I^c$ as the set of indices absent from $I$, identified with the corresponding increasing sequence. We will write $\epsilon(I)$ for the signature of the permutation which has $I$ as its 1st $p$ values and then $I^c$ (the permutation is a shuffle in the case $I$ follows the increasing order) so that
\begin{equation}\label{eqno:eIIcvol}
	e^I \wedge e^{I^c} = \epsilon(I)\vol
\end{equation}
We can then express the contraction in the following way
\begin{equation}\label{eqno:eInp}
	e^{(n-p)}_I = u_I\lr \vol = u_I\lr \epsilon(I)e^I\wedge e^{I^c}
	 = (u_I\lr e^I)\epsilon(I)\wedge e^{I^c} = \epsilon(I)e^{I^c}
\end{equation}
Note that this formula can act as a definition in the case $(e^i)$ form a general family of vectors (not assumed to be a basis).
Combining (\ref{eqno:eIIcvol},\ref{eqno:eInp}) we obtain the following formula stating that $e^{(n-p)}_I$ represent a (twisted) $1$-form on $\ExT^p E$ : for $f_Je^J\in\ExT^p E$ 
\begin{equation}\boxed{
	\lp f_Je^J \rp \wedge e^{(n-p)}_I = f_I\vol}
\end{equation}
which justifies placing $I$ as a subscript in $e^{(n-p)}_I$.





\renewcommand{\l}{\lbar}
\emergencystretch=1em
\printbibliography[heading=bibintoc]

\end{document}